\documentclass[a4paper,UKenglish,cleveref,autoref,thm-restate]{lipics-v2021}

\pdfoutput=1 
\hideLIPIcs  

\title{Semantic foundations of equality saturation}

\author{Dan Suciu}{University of Washington, Seattle}{suciu@cs.washington.edu}{https://orcid.org/0000-0002-4144-0868}{}

\author{Yisu Remy Wang}{University of California, Los Angeles}{remywang@cs.ucla.edu}{https://orcid.org/0000-0002-6887-9395}{}

\author{Yihong Zhang}{University of Washington, Seattle}{yz489@cs.washington.edu}{https://orcid.org/0009-0006-5928-4396}{}

\authorrunning{D. Suciu, Y.\,R. Wang, and Y. Zhang} 

\Copyright{Dan Suciu, Yisu Remy Wang, and Yihong Zhang} 



\category{} 

\relatedversion{} 

\acknowledgements{This work is supported by NSF IIS 2314527 and NSF SHF 2312195.}

\nolinenumbers 

\EventEditors{John Q. Open and Joan R. Access}
\EventNoEds{2}
\EventLongTitle{42nd Conference on Very Important Topics (CVIT 2016)}
\EventShortTitle{CVIT 2016}
\EventAcronym{CVIT}
\EventYear{2016}
\EventDate{December 24--27, 2016}
\EventLocation{Little Whinging, United Kingdom}
\EventLogo{}
\SeriesVolume{42}
\ArticleNo{23}
\addto\extrasUKenglish{

}

\usepackage{mathtools}
\usepackage{algpseudocode}
\usepackage{algorithm}
\usepackage{amsmath}
\usepackage{wrapfig}
\usepackage{subcaption}
\usepackage{tikz-cd}
\usepackage{xspace}
\usepackage{xcolor}

\usepackage{tensor}
\usepackage[normalem]{ulem} 
\usepackage[textsize=small]{todonotes}

\declaretheorem[name=Theorem]{thm}
\declaretheorem[name=Lemma,sibling=thm]{lem}
\declaretheorem[name=Definition,sibling=thm]{dfn}
\declaretheorem[name=Corollary,sibling=thm]{cor}
\declaretheorem[name=Example,sibling=thm]{exm}
\declaretheorem[name=Fact,sibling=thm]{fac}

\begin{document}

\newcommand{\egraphs}{\mbox{E-graphs}\xspace}
\newcommand{\egraph}{\mbox{E-graph}\xspace}
\newcommand{\Egraph}{\mbox{E-graph}\xspace}
\newcommand{\Egraphs}{\mbox{E-graphs}\xspace}
\newcommand{\eclass}{\mbox{E-class}\xspace}
\newcommand{\Eclass}{\mbox{E-class}\xspace}
\newcommand{\Enodes}{\mbox{E-nodes}\xspace}
\newcommand{\enodes}{\mbox{E-nodes}\xspace}
\newcommand{\enode}{\mbox{E-node}\xspace}
\newcommand{\Eclasses}{\mbox{E-classes}\xspace}
\newcommand{\eclasses}{\mbox{E-classes}\xspace}
\newcommand{\ematch}{\mbox{e-match}\xspace}
\newcommand{\ematching}{\mbox{E-matching}\xspace}
\newcommand{\Ematching}{\mbox{E-matching}\xspace}
\newcommand{\cclyzerpp}{\texttt{cclyzer++}\xspace}
\newcommand{\souffle}{Souffl\'e\xspace}

\newcommand{\egg}{\mbox{\texttt{egg}}\xspace}
\newcommand{\egglog}{\mbox{\texttt{egglog}}\xspace}
\newcommand{\Egglog}{\mbox{\texttt{egglog}}\xspace}

\newcommand{\equivid}{\ensuremath{\equiv_{\sf id}}\xspace}
\newcommand{\find}{\textsf{find}\xspace}
\newcommand{\lookup}{\textsf{lookup}\xspace}

\newcommand{\set}[1]{\{#1\}}                    
\newcommand{\setof}[2]{\{{#1}\mid{#2}\}}        
\newcommand{\N}{{\mathbb N}}    			
\newcommand{\R}{{\mathbb R}}    			
\newcommand{\Z}{{\mathbb Z}}    			
\newcommand{\defeq}{\stackrel{\text{def}}{=}}
\newcommand{\dom}{\textsf{Dom}}
\newcommand{\arity}{\textit{ar}}

\newcommand{\chase}{\textsc{Ch}}
\newcommand{\sklchase}{\textsc{SklCh}}
\newcommand{\stdchase}{\textsc{StdCh}}
\newcommand{\eqsat}{\textsc{EqSat}}
\newcommand{\congr}{\textsf{CC}}
\newcommand{\varset}{\textit{\textsf{Var}}}
\newcommand{\lhs}{\textit{lhs}}
\newcommand{\rhs}{\textit{rhs}}

\newcommand{\trs}{{\mathcal R}}

\newenvironment{proofsketch}{%
  \renewcommand{\proofname}{Proof sketch}\proof}{\endproof}

\newcommand{\flatt}{\textsf{FL}}

\newcommand{\calA}{\mathcal A}
\newcommand{\calB}{\mathcal B}
\newcommand{\calC}{\mathcal C}
\newcommand{\calH}{\mathcal H}
\newcommand{\calE}{\mathcal E}
\newcommand{\calD}{\mathcal D}
\newcommand{\calF}{\mathcal F}
\newcommand{\calL}{\mathcal L}
\newcommand{\calS}{\mathcal S}
\newcommand{\calT}{\mathcal T}
\newcommand{\calZ}{\mathcal Z}
\newcommand{\calV}{\mathcal V}
\newcommand{\calW}{\mathcal W}
\newcommand{\ico}{\textsf{ICO}}

\newcommand{\mysubparagraph}[1]{\vspace{-0.5em}\subparagraph{#1}}
\newcommand{\supp}{\textit{supp}}

\newcommand{\revinsA}[1]{{#1}}
\newcommand{\revdelA}[1]{}
\newcommand{\revinsB}[1]{{#1}}
\newcommand{\revdelB}[1]{}
\newcommand{\revinsC}[1]{{#1}}
\newcommand{\revdelC}[1]{}
\newcommand{\revinsAll}[1]{{#1}}
\newcommand{\revdelAll}[1]{}

\maketitle

\begin{abstract}
  Equality saturation is an emerging technique for program and query optimization
  developed in the programming language community.
  It performs term rewriting over an E-graph,
  a data structure that compactly represents a program space.
  Despite its popularity,
  the theory of equality saturation lags behind the practice.
  In this paper,
  we define a fixpoint semantics of equality saturation based on tree automata
  and uncover deep connections between equality saturation and the chase.
  We characterize the class of chase sequences that correspond to equality saturation.
  We study the complexities of terminations of equality saturation
  in three cases: single-instance, all-term-instance, and all-E-graph-instance.
  Finally, we define a syntactic criterion based on acyclicity that
  implies equality saturation termination.
\end{abstract}

\section{Introduction}

\label{sec:intro} 

\revdelA{
An E-graph is a data structure introduced in
 the programming language and formal verification community in the 1970s
 for efficiently answering the word problem of ground terms---given 
 a set of identities between ground terms 
 and two ground terms $u$ and $v$,
 an E-graph can efficiently answer whether the given identities
 imply $u\approx v$.
The idea
 of using E-graphs for program optimization, 
 known as \emph{equality saturation (EqSat)}, is 
 first explored in 2000s.
More recently,
 Willsey et al. introduced \egg,
 a high-performance, generic EqSat library.
This project has had great impacts and
 led to a surge of interests in 
 using EqSat for program/query optimization.
Since 2021,
 there have been dozens of projects, both in academia and in industry,
 that use equality saturation.
These projects cover a wide range of topics in domain-specific program optimization,
 including 
 floating-point computation
 computational fabrication,
 machine learning systems, 
 and hardware design---just to name a few.
There is also a growing interest in using EqSat
 for query optimization in data management.
For example,
 EqSat is used to optimize queries in OLAP,
 linear algebra,
 tensor algebra,
  and Datalog settings.
}

\revdelA{
  In a nutshell, an E-graph is a compact representation of a (possibly
  infinite) set of ground terms.  Equality saturation processes each
  identity $u \approx v$ by matching the term $u$ with the E-graph,
  then adding the term $v$ to the E-graph.  There are striking
  connections between equality saturation and database concepts.
  Zhang et al. observed that {\em matching} is
  the same as conjunctive query evaluation, and described significant
  speedups by using a Worst Case Optimal Join
  algorithm for matching.
  Recently, the \Egglog system unified EqSat and
  Datalog, allowing it to extend capabilities of \egg for program
  optimization and program analysis tasks.
}

\revinsA{Given a set of identities between terms, the word problem
  asks whether the identities imply two ground terms $t_1, t_2$
  are equivalent, i.e. $t_1\approx t_2$.  This fundamental problem 
  has applications including automated theorem proving, program
  verification, and query equivalence checking.  In his Ph.D.\
  thesis, Nelson~\cite{nelson} introduced a data structured called {\em
    E-graph} for efficiently answering the word problem.
  At the core, an E-graph is a compact representation of an
  equivalence relation over a possibly infinite set of ground terms.
  During the 2000s, researchers applied E-graphs to program
  optimization~\cite{denali, eqsat}. The compiler populates an E-graph with many
  equivalent programs, using axiomatic rewrites, then extracts the
  best program from the equivalent ones.  
  In particular, Tate et.al.~\cite{eqsat} coined the term \emph{equality
    saturation (EqSat)} and gave a procedural description of the
  algorithm.  In 2021, Willsey et al.~\cite{egg} proposed \egg, which
  introduced important algorithmic improvements and made EqSat
  practical.
  Since 2021, EqSat has been applied to a
  wide range of topics in domain-specific program optimization,
  including floating-point computation~\cite{herbie} computational
  fabrication~\cite{szalinski}, machine learning
  systems~\cite{tensat}, and hardware design
  \cite{diospyros,Coward2023AutomatingCD}.  There is also a growing
  interest in using EqSat for query optimization in data management.
  For example, EqSat is used to optimize queries in
  OLAP~\cite{risinglight}, linear algebra~\cite{spores}, tensor
  algebra~\cite{storel}, and Datalog~\cite{wang2022optimizing}.  }

\revinsA{The equality saturation procedure consists of
  repeatedly selecting an identity $u \approx v$ from the given set,
  matching the term $u$ with the E-graph, then adding the term $v$ to
  the E-graph, if it wasn't already there.  Equality saturation
  terminates when no new terms can be added.  There are striking
  connections between equality saturation and database concepts.
  Zhang et al.~\cite{relational-ematching} observed that the {\em
    matching} step is the same as conjunctive query evaluation, and
  described significant speedups in \egg by using a Worst Case Optimal
  Join algorithm~\cite{DBLP:journals/sigmod/NgoRR13} for matching.  A
  recent system, \Egglog~\cite{egglog}, unified EqSat and Datalog
  to improve \egg's support for program optimization
  and program analysis.  }

  In this paper we study another deep connection between equality
  saturation and the chase procedure for Tuple Generating Dependencies
  (TGDs) and Equality Generating Dependencies
  (EGDs)~\cite{FAGIN200589}.  Our hope is that these results will help
  solve some of the open problems in equality saturation by using
  techniques and results for the chase procedure.  Before describing
  our results we give a gentle introduction to EqSat and
  describe some of its open problems.

  \revdelA{ Equality saturation starts by constructing a tree
    automaton (known in the EqSat community as \emph{E-graphs})
    accepting exactly the program to be optimized, as shown on the
    left in Figure 1.  Then, it applies identities to the tree
    automaton, expanding it to represent more and more equivalent
    programs.  To apply an identity, we first match the left-hand side
    against states in the tree automaton.  We will define formally
    what it means for a state to match a pattern in~Sec.~3, but for
    our example, we see that the pattern $f(x, y)$ matches the states
    $c_2$, $c_3$, and $c_4$.  For every match, we instantiate the
    pattern variables with the children states of the match, and
    insert the right-hand side of the identity into the tree
    automaton.  The right side of Figure 1 shows the result of
    inserting $g(x, y)$ for every match to the automaton.  At this
    point, we see that the automaton now accepts any term of the form:
    $h(h(h(a, a),h(a, a)), h(h(a, a),h(a, a)))$
    where $h$ is either $f$ or $g$, totaling $2^7$ terms.  We repeat
    matching and applying identities until either reaching a fixpoint
    (thereby ``saturating'' the automaton), or exhausting some given
    time or memory budget.  At the end of this process, we will have
    grown the automaton to represent a large number of equivalent
    programs.
    The final step is to extract the best program from the automaton
    according to some cost model (e.g., using Knuth's algorithm).  }
\begin{exm} \label{ex:intro:f:g}
  \revinsA{Consider a simple language with two binary operators $f,g$ 
    and constant $a$. We want to optimize the following term
    $t$ (the ``8th power'' of $f$ on $a$):
    \begin{align}
      t = & f(f(f(a, a),f(a, a)), f(f(a, a),f(a, a))) \label{eq:intro:t}
    \end{align}
    We are given a single identity, $f(x, x) \approx g(x, x)$, which
    says two terms $f(t_1,t_2)$ and
    $g(t_1,t_2)$ are equivalent, provided that $t_1, t_2$ are
    equivalent.  
    Starting with the initial term $t$, 
    EqSat constructs an E-graph $G$ and grows it
    to represent all terms equivalent to $t$. 
    The literature defines an E-graph as a set of
    {\em E-classes}, where each E-class is a set of {\em E-nodes}, and
    each E-node is labeled with a function symbol and has a number of
    E-class children equal to the arity of the symbol. EqSat
    starts by constructing an E-graph $G$ representing $t$, shown on the left in
    \autoref{fig:egraph}.  There are 4 E-classes, each consisting of
    one single E-node; the E-class $c_4$ represents precisely the term
    $t$.  Next, EqSat repeatedly applies the identity
    $f(x,x) \approx g(x,x)$, by matching the left-hand side $f(x,x)$ to the
    E-graph, then adding the right-hand side $g(x,x)$ to the E-graph: we formalize
    this in~\autoref{sec:defintion}.  The resulting E-graph $H$ is on
    right of Figure~\ref{fig:egraph}.  There are 4 E-classes,
    $c_1, \ldots, c_4$, each consisting of 1 or 2 E-nodes.  For
    example, $c_4$ has two E-nodes, and represents two equivalent
    terms, $f(t_1, t_2)\approx g(t_1,t_2)$, where $t_1, t_2$ are any
    terms represented by $c_3$.  By continuing this reasoning, we
    observe that $c_4$, represents a total of $2^7$ possible terms,
    namely all terms of the form:
    $h(h(h(a, a),h(a, a)), h(h(a, a),h(a, a)))$
    where each $h$ can be either $f$ or $g$.
  }
\end{exm}

\begin{wrapfigure}{r}{0.3\textwidth}
  \vspace{-1cm}
    \centering
    \begin{subfigure}[b]{0.3\linewidth}
      \centerline{$G:$}
      \centerline{\includegraphics[height=4cm]{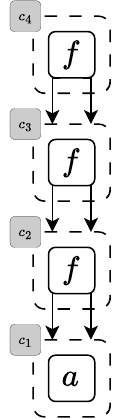}}
    \end{subfigure}
    \hfill
    \begin{subfigure}[b]{0.65\linewidth}
      \centerline{$H:$}
      \centerline{\includegraphics[height=4cm]{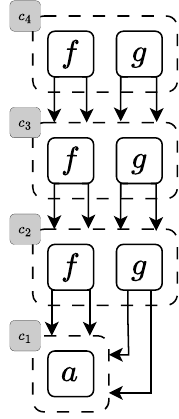}}
    \end{subfigure}
    \caption{\revinsA{Two E-graphs $G, H$, before and after EqSat.}
        \revdelA{Example automata used in equality saturation.
        Each state is shown by a dotted box with label.  The transition
        $f(c_1, c_1) \rightarrow c_2$ is shown by the solid box
        $f(c_1, c_1)$ in the E-class $c_2$.}}
    \label{fig:egraph}
  \vspace{-0.5cm}
\end{wrapfigure}

\mysubparagraph{Open problems about EqSat}
We still understand very little about equality saturation.
Most descriptions of EqSat focus on 
 an imperative understanding\footnote{
  A notable exception is \texttt{egglog}~\cite{egglog},~whose semantics is based on fixpoints
  instead of implementation details.
  Some early works also define E-graphs (under a different name like abstract congruence closure) 
  as tree automata similar to ours~\cite{snyder1993rgrs,abstract-congruence-closure,gulwani2005join}.
 } of equality saturation and E-graphs.
E-graphs are described by their individual components (e.g., a hashconsing data structure,
 a union-find, etc.),
 and EqSat is commonly defined in pseudocode as a sequence of operations.
\revdelA{
 As a result, they cannot talk about the semantics of EqSat in cases EqSat does not terminate,
 and fall short of describing the properties that EqSat enjoys.}
\revinsA{In other words, the semantics of EqSat is the output of the
  specific algorithm, if it terminates; if the algorithm diverges, the
  semantics is undefined.}

We also do not know much about when  EqSat  terminates.
The termination problem  asks, given a set of rewrite rules,
 whether EqSat terminates on  a given input E-graph, or whether it terminates
 on all input E-graphs.
This is a fundamental problem of EqSat and has applications in program/query optimization
 and equivalence checking:
If EqSat terminates on the symmetric closure of a set of (variable-preserving) rewrite rules~$\trs$,
 it decides the word problem of the equational theory defined by $\trs$ (\autoref{lemma:representation}).
With an appropriate cost model,
 EqSat can further pick the optimal program among all programs
 equivalent to the input, e.g. by using Knuth's
 algorithm~\cite{DBLP:journals/ipl/Knuth77}.

 \mysubparagraph{Our contribution}
 \revdelA{We first introduce in Section~3 a new formal definition of
   E-graphs and equality saturation, based on tree automata: }
 \revinsA{After a review of some background material
   in~\autoref{sec:background}, we introduce E-graphs and EqSat
   in~\autoref{sec:defintion}.  Our definition,
   in~\autoref{subsec:egraph}, applies even to the cases when equality
   saturation does not terminate and, for that purpose, we define the
   E-graph to be a reachable, deterministic tree automaton, with
   possibly infinitely many states.  By explicitly allowing infinite
   E-graphs we can define a formal semantics even when EqSat does not
   terminate.  We show that concepts in tree automata are in 1-1
   correspondence with those in E-graphs: the automaton states
   correspond to E-classes, and the transitions correspond to
   E-nodes.}  A term is represented by an E-graph iff it is accepted
 by the E-graph viewed as a standard tree automaton.
%
\revinsA{We  prove that, for any
two E-graphs there exists at most one homomorphism between them, and,
therefore, E-graphs are \emph{rigid} tree automata.}
Next, in~\autoref{sec:egraph-operations}, we define a few basic
operations on E-graphs, such as E-matching, insertion, congruence
closure, and least upper bounds, by relying on tree-automata concepts.
%
\revinsB{Using these operations, we define in~\autoref{subsec:eqsat}}
an \emph{immediate consequence operator (ICO)}, and define EqSat 
formally as the least fixpoint of the ICO.  The least
fixpoint always exists and is unique, \revinsA{even if the fixpoint
  procedure does not terminate, in which case the least fixpoint may
  be infinite.}
%
Finally, we prove an important lemma, called the \emph{Finite
  Convergence Lemma}, stating that, if the least fixpoint is
finite, then equality saturation procedure converges in
\revdelC{finite}\revinsC{finitely many} steps.  This is not
immediately obvious because, while E-matching and insertion strictly
increase the size of the E-graph, congruence closure may decrease it.
A similar proposition fails for TGDs and EGDs: there exists an infinite
chase where all instances have bounded size, hence its ``limit'' is
finite.
%

Next, in \autoref{sec:eqsat:and:chase} we describe the connection
between EqSat and the chase.  \revinsA{After a brief review of the
  chase in~\autoref{sec:chase}}, we start by 
presenting a reduction from the Skolem chase to equality saturation,
denoted $\sklchase\Rightarrow \eqsat$ (\autoref{sec:chase-to-eqsat}), then
from equality saturation to the standard chase, denoted $\eqsat \Rightarrow
\stdchase$ (\autoref{sec:eqsat-to-chase}).
For  $\sklchase\Rightarrow \eqsat$,
 given a set of dependencies, 
 we show there exists a set of rewrite rules 
 where EqSat produces an encoded result of the Skolem chase and has the same termination behavior.
For $\eqsat \Rightarrow \stdchase$,
 we show \revinsC{that,} given a set of rewrite rules,
 there exists a set of dependencies 
 where the standard chase produces an encoded result of EqSat (whether it terminates).
Since the standard chase is a non-deterministic process,
 we characterize the type of chase sequences that terminate
 when EqSat terminates (\autoref{thm:eqsat-to-chase}).
We call them \emph{EGD-fair} chase sequences.
Intuitively, a chase sequence is called EGD-fair if it applies EGDs to a fixpoint frequently enough.
We show in \autoref{thm:eqsat-to-chase} that,

\begin{tabular}{rl}
  EqSat terminates &$\Leftrightarrow$ one chase sequence terminates\\
   &$\Leftrightarrow$
 all EGD-fair chase sequences terminate.
\end{tabular}

\noindent
The notion of EGD-fair chase sequences is of independent interest.

Finally, we present our main decidability results for EqSat in
\autoref{sec:terminations}: we show that the single-instance
termination problem of EqSat, denoted as $\mathcal{T}^\eqsat_G$, is
R.E.-complete, and the all-term-instance termination problem of EqSat,
denoted as $\mathcal{T}^\eqsat_{\forall t}$, is $\Pi_2$-complete.  Our
proof is based on a non-trivial reduction from the Turing machine,
first presented in the undecidability proof of the finiteness of
congruence classes defined by string rewriting
systems~\cite{narendran1985complexity}.  While the single-instance
case easily follows from the undecidability of Skolem chase
termination, our approach allows us to also prove the
$\Pi_2$-completeness of the all-term-instance termination case by a
reduction from the universal halting problem.  We also show the
all-E-graph-instance termination problem of EqSat, denoted as
$\mathcal{T}^\eqsat_{\forall G}$, is undecidable, although the exact
upper bound is open.

We contrast the termination problems of EqSat with those of the Skolem chase and the standard chase.
The single-instance termination problems are R.E.-complete in all three cases~\cite{marnette2009generalized,chase-revisited},
 and the all-instance termination of the Skolem chase ($\mathcal{T}^\sklchase_\forall$) is R.E.-complete 
 as well~\cite{marnette2009generalized,gogacz2014allinstance}.
This shows that $\mathcal{T}^\sklchase_\forall$ is easier than $\mathcal{T}^\eqsat_{\forall t}$.
The case for the standard chase is more interesting.
There are two all-instance termination problems
 of the standard chase: for all database instances, whether all chase sequences terminate
 in \revdelC{finite}\revinsC{finitely many} steps ($\mathcal{T}^\stdchase_{\forall,\forall}$), 
 and whether there exists at least one chase sequence that terminate 
 ($\mathcal{T}^\stdchase_{\forall,\exists}$).
It has been shown
 $\mathcal{T}^\stdchase_{\forall,\exists}$ is $\Pi_2$-complete~\cite{grahne18anatomy},
 but the exact complexity of $\mathcal{T}^\stdchase_{\forall,\forall}$ is open.
Grahne and Onet showed if we allow one \emph{denial constraint},
 $\mathcal{T}^\stdchase_{\forall,\forall}$ is $\Pi_2$-complete~\cite{grahne18anatomy},
 although Gogacz~and~Marcinkowski~\cite{gogacz2014allinstance} conjectured
 that this problem is indeed in R.E.



 In \autoref{sec:acyclicity} we propose a sufficient syntactic criterion 
 that guarantees EqSat termination, called \emph{weak term
   acyclicity}, which is based on the classic notion of \emph{weak
   acyclicity}~\cite{FAGIN200589}.  If a set of rewrite rules is
 weakly term acyclic, then EqSat terminates for all input E-graphs.




\section{Background}
\label{sec:background}

\revdelA{The section on the backgrounds of the chase procedure is moved to Section 4.}

\subsection{Term Rewriting Systems}

\label{eq:trs}

We review briefly the standard definition of a term rewriting system
from~\cite{traat}.  A \emph{signature} is a finite set $\Sigma$ of
function symbols with given arities.  If $V$ is a set of variables,
then $T(\Sigma, V)$ denotes the set of terms constructed inductively
using symbols from $\Sigma$ and variables from $V$.  Members of
$T(\Sigma, V)$ are called \emph{patterns}, and members of
$T(\Sigma) \defeq T(\Sigma, \emptyset)$ are called \emph{ground
  terms}, or simply \emph{terms} thereafter.  A \emph{substitution} is
a function $\sigma : V \rightarrow T(\Sigma)$; if $u$ is a pattern,
then we denote by $u[\sigma]$ the term obtained by applying the
substitution $\sigma$ to $u$.  A \emph{rewrite rule} $r$ has the form
$\lhs\rightarrow\rhs$ where \lhs{} and \rhs{} are patterns and the
variables in \rhs{} are a subset of those \lhs{},
$\varset(\rhs)\subseteq\varset(\lhs)$.  A \emph{term rewriting system}
(TRS), $\trs$, is a set of rewrite rules.  $\trs$ defines a
\emph{rewrite relation} $\rightarrow_{\trs}$ as follows:
$\lhs[\sigma] \rightarrow_{\trs} \rhs[\sigma]$ for any substitution
$\sigma$ and rule $\lhs\rightarrow \rhs$ in $\trs$, and, if
$u \rightarrow_{\trs} v$ then
$f(w_1, \ldots, w_{i-1},u,w_{i+1},\ldots w_k)\rightarrow_{\trs}f(w_1,
\ldots, w_{i-1},v,w_{i+1},\ldots w_k)$ for any function symbol
$f \in \Sigma$ of arity $k$, and any terms $w_j$, $j=1,k; j\neq i$.
%
%
Let $\rightarrow_{\trs}^*$ be the reflexive and transitive closure of
$\rightarrow_{\trs}$.  We define
$(\leftarrow_{\trs})\defeq(\rightarrow_{\trs})^{-1}$,
$(\leftrightarrow_{\trs})\defeq(\rightarrow_{\trs})\cup
(\leftarrow_{\trs})$, and
$(\approx_{\trs})\defeq(\leftrightarrow_{\trs}^*)$.  $\approx_{\trs}$
is a congruence relation.  We define the set of reachable terms
$R^*(t)=\{t' \mid t\rightarrow_{\trs}^* t'\}$.  If a term rewriting
system $\trs$ is variable-preserving (i.e.,
$\varset(\lhs)=\varset(\rhs)$ for all rules), we define
$\trs^{-1}=\setof{\rhs\rightarrow\lhs}{(\lhs\rightarrow\rhs)\in
  \trs}$.  It follows that
$(\rightarrow_{\left(\trs^{-1}\right)})=(\leftarrow_{\trs})$.


\subsection{Tree automata}
\label{sec:tree-automata}

Let $\Sigma$ be a signature.  A (bottom-up) tree automaton is a tuple
$\mathcal{A}=\langle Q, \Sigma, Q_{\textit{final}}, \Delta\rangle$,
where $Q$ is a (potentially infinite)\footnote{ In this paper, we
  allow tree automata (and thus E-graphs) to have an infinite number
  of states and transitions.  Talking about infinite E-graphs allow us
  to define the semantics of equality saturation even when the
  algorithm does not terminate.  } set of states,
$Q_{\textit{final}}\subseteq Q$ is a set of final states, and
$\Delta$ is a set of transitions of the form
$f(q_1,\ldots, q_n)\rightarrow q$ where $q,q_1,\ldots, q_n\in Q$, and
$f \in \Sigma$.  Denote by $\Sigma \cup Q$ the signature $\Sigma$
extended with $Q$ where each state is viewed as a symbol of arity
0. Then $\Delta$ is a term rewriting system for $\Sigma \cup Q$, and
we will denote by $\rightarrow_{\calA}^*$ (rather than
$\rightarrow_\Delta^*$) the rewrite relation defined by $\Delta$.  A
term $t \revinsC{\in T(\Sigma)}$ is accepted by a state $q$ if
$t\rightarrow_{\calA}^* q$, and we write $\calL(q)$ for the set of
terms accepted by $q$.  The language accepted by $\calA$ is
$\calL(\calA) \defeq \setof{t \revinsC{\in T(\Sigma)}}{\exists
  q_{f}\in Q_{\textit{final}},\,t \rightarrow^*_{\calA} q_{f}}$.  A
tree language $L\subseteq T(\Sigma)$ is called regular if it is
accepted by some \emph{finite} tree automaton.

Fix two tree automata
$\calA = \langle Q, \Sigma, Q_{\textit{final}},\Delta\rangle, \calB=\langle Q', \Sigma,
Q'_{\textit{final}},\Delta'\rangle$.  A \emph{homomorphism},
$h :\calA \rightarrow \calB$, is a function $h:Q \rightarrow Q'$ that
maps final states to final states, and, for every transition
$f(c_1,\ldots,c_k)\rightarrow c$ in $\calA$ there exists a transition
$f(h(c_1),\ldots, h(c_k))\rightarrow h(c)$ in $\calB$.  An
\emph{isomorphism}\footnote{Notice
that a bijective homomorphism is not necessarily an isomorphism.} is a homomorphism $h : \calA \rightarrow \calB$ for
which there exists an inverse homomorphism
$h^{-1}: \calB \rightarrow \calA$ such that  $h^{-1}\circ h = id_{\calA}$
and $h\circ h^{-1} = id_{\calB}$.
 The following holds:

\begin{lem} \label{lemma:homomoprhism:simple} 
  \revdelA{If
  $h : \calA \rightarrow \calB$ is a homomorphism and
  $t \rightarrow^*_{\calA} c$, then $t \rightarrow^*_{\calB} h(c)$.}
  \revinsA{
    Let $h : \calA \rightarrow \calB$ be a homomorphism,
    $t\in T(\Sigma)$, and $c$ be a state of $\calA$.
    If $t \rightarrow^*_{\calA} c$, then $t \rightarrow^*_{\calB} h(c)$.
  }
  In particular, $\calL(\calA) \subseteq \calL(\calB)$.
\end{lem}

\begin{proof}
  We prove the statement by induction on the structure of the term
  $t \in T(\Sigma)$.  Assuming $t = f(t_1, \ldots, t_k)$ for $k\geq 0$\footnote{
    \revinsB{The base case is covered by the case $k=0$.}}
  and $t\rightarrow^*_{\calA} c$, then there exists states
  $c_1, \ldots, c_k$ such that $t_i \rightarrow^*_{\calA} c_i$ and
  $\calA$ contains the transition $f(c_1, \ldots, c_k) \rightarrow c$.
  By induction hypothesis $t_i \rightarrow^*_{\calB} h(c_i)$ for
  $i=1,\ldots, k$, and since $h$ is a homomorphism, there exists a transition
  $f(h(c_1),\ldots,h(c_k)) \rightarrow h(c)$ in $\calB$, proving that
  $t \rightarrow^*_{\calB} h(c)$.
\end{proof}

We write $\calA \sqsubseteq \calB$ whenever there exists a homomorphism
$\calA \rightarrow \calB$.  Observe that $\sqsubseteq$ is a preorder
relation.
\revinsA{In the next section, we show that 
 this preorder relation $\sqsubseteq$ becomes a partial order when restricted to E-graphs (\autoref{lem:egraph-partial-order}).}

We call an automaton $\calA$ \emph{deterministic} if
$t\rightarrow^*_{\calA} q_1$ and $t\rightarrow^*_{\calA} q_2$ implies
$q_1=q_2$ for \revinsA{states $q_1,q_2$}.  We call $\calA$ \emph{reachable} if every state $q$
accepts some ground term: $\exists t \in T(\Sigma)$,
$t\rightarrow^*_{\calA} q$.

\section{E-graphs and Equality Saturation}
\label{sec:defintion}

Most papers discussing E-graphs use an operational definition not 
suitable for a theoretical analysis.  We introduce an
equivalent definition of E-graphs in terms of tree
automata, similar to Kozen's partial tree automata~\cite{kozen1993partial}.
Throughout this section we fix the signature $\Sigma$.

\subsection{E-graphs}

\label{subsec:egraph}

\begin{dfn} \label{def:egraph} An \emph{E-graph} is a deterministic
  and reachable tree automaton $G = \langle Q, \Sigma, \Delta\rangle$
  (without a set of final state $Q_{\textit{final}}$).
\end{dfn}

\revdelA{ The convention in the E-graphs literature is to call a state
  $q \in Q$ an \emph{E-class} and a transition
  $f(q_1,\ldots, q_n)\rightarrow q$ an \emph{E-node}; furthermore, the
  set of E-classes and E-nodes are denoted $C$ and $N$ respectively.
  We use states/E-classes and transitions/E-nodes interchangeably in
  this paper.  For most of our discussion we could ignore the final
  states $Q_{\textit{final}}$ (as done in~\cite{kozen1993partial}); we
  included them in Definition 2 to keep the analogy with tree
  automata.  If a term $t$ is accepted by $c$ (i.e.
  $t \rightarrow^*_G c$) then we say that $t$ is \emph{represented} by
  the E-class $c$.
%
}

\revinsA{ Our definition maps one-to-one to the classical definition
  of E-graphs: An E-class is a state $c \in Q$ of the tree automaton,
  and an E-node is a transition $f(c_1,\ldots,c_k) \rightarrow c$.  A
  term $t$ is \emph{represented} by the E-class $c$ if $t$ is accepted
  by $c$, i.e.  $t \rightarrow^*_G c$.  In the literature, the sets of
  E-classes and E-nodes are denoted $C$ and $N$ respectively.  We will
  use states/E-classes and transitions/E-nodes interchangeably in this
  paper.  E-graphs do not define a set of ``final'' E-classes, and for
  that reason we omit the final states $Q_{\textit{final}}$
  from~\autoref{def:egraph}\footnote{
    Alternatively, consider $Q_{\textit{final}}=Q$.
  }, similarly to~\cite{kozen1993partial}.
  %
}

\revinsA{
\begin{exm}
  The E-graph $H$ in \autoref{fig:egraph} is the automaton
  $\langle Q, \Sigma, \Delta\rangle$, where
  $\Sigma = \set{a,f(\cdot, \cdot),g(\cdot, \cdot)}$, there are four
  states $Q = \set{c_1,\ldots,c_4}$, and $\Delta$ consists of seven
  transitions:
    \begin{align*}
      a() \rightarrow & c_1 & f(c_1,c_1) \rightarrow & c_2 & g(c_1,c_1) \rightarrow & c_2
&&\ldots &f(c_3,c_3) \rightarrow & c_4 & g(c_3,c_3) \rightarrow & c_4
    \end{align*}
    An example of rewritings is
    $f(a,a) \rightarrow_H f(c_1,a) \rightarrow_H f(c_1,c_1)
    \rightarrow_H c_2$, showing that the term $f(a,a)$ is represented
    by the E-class $c_2$.
\end{exm}
}

\revinsA{It is folklore that E-graphs represent equivalences of
  terms.  We make this observation formal, by defining the semantics
  of an E-graph to be a certain partial congruence.}  A {\em partial
  equivalence relation}, or PER, on a set $A$ is a binary relation
$\approx$ that is symmetric and transitive.  Its {\em support} is the
set $\supp(\approx) \defeq \setof{x}{x\approx x}\subseteq A$.
Equivalently, a PER can be described by its support and an equivalence
relation on the support.
A PER on the set of terms $T(\Sigma)$ is \emph{congruent} if
$s_i\approx t_i$ for $i=1,\ldots,n$ and
$f(s_1,\ldots, s_n)\in \supp(\approx)$ implies
$f(s_1, \ldots, s_n)\approx f(t_1,\ldots,t_n)$.  A PER is reachable if
$f(s_1, \ldots, s_n) \in \supp(\approx)$ implies
$s_i \in \supp(\approx)$, for $i=1,\ldots,n$.  A {\em Partial
  Congruence Relation (PCR)}\footnote{PCRs are studied in the
  literature as congruences on partial algebras
  (e.g.,~\cite{kozen1993partial}).} on $T(\Sigma)$ is a congruent and
reachable PER.

An E-graph $G$ induces a PCR $\approx_G$ defined as follows:
$t_1 \approx t_2$ if there exists some E-class $c$ in $G$ that accepts
both $t_1$ and $t_2$, i.e.
$t_1\rightarrow^*_G c \tensor*[^*_G]{\leftarrow}{} t_2$.  We check
that $\approx_G$ is a PCR:
%
$\approx_G$ is symmetric by definition, and transitivity follows from
determinacy, because
$t_1\rightarrow^*_G c \tensor*[^*_G]{\leftarrow}{} t_2$ and
$t_2\rightarrow^*_G c' \tensor*[^*_G]{\leftarrow}{} t_3$ implies
$t_1\rightarrow^*_G c=c' \tensor*[^*_G]{\leftarrow}{} t_3$.  Suppose
$f(s_1, \ldots, s_n) \rightarrow^*_G c$: then there exists states
$c_i$ s.t.  $s_i \rightarrow^*_G c_i$, and a transition
$f(c_1, \ldots, c_n) \rightarrow c$, proving reachability; if, in
addition, $s_i \approx_G t_i$ for $i=1,\ldots, n$, then
$t_i \rightarrow^*_G c_i$, which implies
$f(t_1,\ldots, t_n)\rightarrow^*_G c$, proving congruence,
$f(s_1, \ldots, s_n)\approx_G f(t_1,\ldots,t_n)$.

\begin{dfn}
  The \emph{semantics} of an E-graph $G$ is the PCR $\approx_G$. 
\end{dfn}

\begin{thm}
  For any PCR $\approx$ \revinsA{over $T(\Sigma)$} there exists a
  unique $G$ such that $(\approx_G) = (\approx)$.
\end{thm}

\begin{proofsketch}The states of $G$ are the equivalence classes
  of $\approx$, denoted as $[t]$ for $t \in \supp(\approx)$, and the
  transitions are
  $f([t_1],\ldots,[t_n]) \rightarrow [f(t_1,\ldots,t_n)]$ for all
  $t_1, \ldots, t_n, f(t_1,\ldots,t_n)$ in the support of $\approx$.
  One can check by induction on the size of $t$ that
  $t \in \supp(\approx)$ iff $t \in \supp(\approx_G)$, and
  $t \rightarrow^*_G [s]$ iff $t \approx s$, proving that
  $(\approx_G)=(\approx)$.
\end{proofsketch}


Thus, the semantics of an E-graph $G$ is a PCR $\approx_G$, which is a
congruence on $\calL(G)\defeq \supp(\approx_G)$.  We say that $G$
\emph{represents} the set of terms $\calL(G)$.



\begin{exm} \label{ex:intro:revisited}
  \revdelA{
  Consider the E-graph $G$ on the right of Figure 1, over
  signature $\Sigma = \set{a,f(\cdot, \cdot),g(\cdot, \cdot)}$.  $G$ has 4 states $c_i, i=1,\ldots, 4$ and
  seven transitions, }
\revdelA{
  \begin{align*}
      \revdelA{a() \rightarrow} & \revdelA{c_1} & \revdelA{f(c_1,c_1) \rightarrow} & \revdelA{c_2} & \revdelA{g(c_1,c_1) \rightarrow} & \revdelA{c_2}
&&\revdelA{\ldots} &\revdelA{f(c_3,c_3) \rightarrow} & \revdelA{c_4} & \revdelA{g(c_3,c_3) \rightarrow} & \revdelA{c_4}
  \end{align*}
  }
  \revdelA{ An example of rewritings is
    $f(a,a) \rightarrow_G f(c_1,a) \rightarrow_G f(c_1,c_1)
    \rightarrow_G c_2$; notice that for the purpose of rewriting,
    states are considered nullary symbols.  Some examples of the the
    PCR are $a \approx_G a$ (represented by state $c_1$),
    $f(a,a) \approx_G g(a,a)$ (by state $c_2$),
    $f(f(a,a),g(a,a)) \approx_G g(f(a,a),f(a,a))$ (by state $c_3$),
    etc.}
  \revinsA{ Continuing~\autoref{ex:intro:f:g}, the semantics of the
    E-graph $H$ in \autoref{fig:egraph} is the PCR $\approx_H$ that
    equates $a \approx_H a$ (witnessed by state $c_1$),
    $f(a,a) \approx_H g(a,a)$ (by state $c_2$),
    $f(f(a,a),g(a,a)) \approx_H g(f(a,a),f(a,a))$ (by state $c_3$),
    etc.
    }
\end{exm}

\begin{exm}
  Let $\Sigma=\set{a,f(\cdot)}$.  Consider the E-graph $G$ with a
  single state $c$ and transitions $a() \rightarrow c$,
  $f(c) \rightarrow c$.  It represents infinitely many terms,
  $f^{(k)}(a)$, for $k\geq 0$, and its semantics is the PCR
  $a \approx_G f(a) \approx_G f(f(a)) \approx_G \cdots$
\end{exm}

\begin{exm}
  Let $\Sigma = \set{a,f(\cdot),g(\cdot)}$ and consider the infinite
  E-graph $G$ with states $c,c_0,c_1,c_2,\ldots$ and transitions
  \begin{align*}
    a \rightarrow & c_0 & 
   f(c_i) \rightarrow & c_{i+1}&  g(c_i) \rightarrow & c && i=0,1,2,\ldots
  \end{align*}
  The PCR consists of
  $g(a)\approx_G g(f(a)) \approx_G g(f(f(a))) \approx_G \ldots$,
  defined by the state $c$.  No other distinct terms are in
  $\approx_G$, for example $f(a)\not\approx_G f(f(a))$ because they
  are represented by the distinct states $c_1$ and $c_2$ respectively.
  Although $G$ represents a regular language,
  $\setof{f^{(k)}(a)}{k \geq 0}\cup\setof{g(f^{(k)}(a))}{k \geq 0}$,
  its semantics $\approx_G$ cannot be captured by a finite E-graph.
  This example shows that $\approx_G$ differs from the Myhill-Nerode
  equivalence relation
  \cite{kozen1992myhill}, under which all terms $f^{(k)}(a)$ would be
  equivalent.  It also illustrates the subtle distinction between tree
  automata and E-graphs.  An optimizer that wants to use the identity
  $g(x)=g(f(x))$, but not $x=f(x)$, needs this E-graph to represent
  all terms equivalent to $g(a)$, and cannot use the finite tree
  automaton accepting the regular language $\calL(c)$ because that
  would incorrectly equate all terms $f^{(k)}(a)$.
\end{exm}




Recall the definitions of tree automata homomorphisms in \autoref{sec:tree-automata}.
 When restricted to E-graphs, homomorphisms have some interesting properties:

\begin{lem} \label{lemma:approx:g:h} If $h : G \rightarrow H$ is a
  homomorphism, then
  $\left(\approx_G\right)\subseteq \left(\approx_H\right)$.
\end{lem}

\begin{proof}
  Assume $t_1\approx_G t_2$. Then there exists some E-class $c$ where
  $t_1\rightarrow^*_G c\tensor*[^*_G]{\leftarrow}{} t_2$.  By
  \autoref{lemma:homomoprhism:simple},
  $t_1\rightarrow^*_H h(c)\tensor*[^*_H]{\leftarrow}{} t_2$, implying
  $t_1\approx_H t_2$.
\end{proof}

\begin{lem}
  \label{lem:unique-morphism}
  There exists at most one homomorphism $h : G \rightarrow H$.
\end{lem}

\begin{proof}
  Call the {\em weight} of a state $c$ in $G$ the size of the smallest
  term $t$ such that $t \rightarrow^*_G c$.  Since $G$ is reachable,
  every state has a finite weight.  Given two homomorphisms
  $h_1, h_2: G \rightarrow H$, we prove by induction on the weight of
  $c$ that $h_1(c)=h_2(c)$.  Let $t$ be a term of minimal size such
  that $t \rightarrow^*_G c$, and assume $t = f(t_1, \ldots, t_k)$,
  for $k \geq 0$.  Then there exists states $c_1, \ldots, c_k$ such
  that $t_i \rightarrow^*_G c_i$, $i=1,k$, and a transition
  $f(c_1, \ldots, c_k) \rightarrow c$ in $G$. By induction hypothesis
  $h_1(c_i)=h_2(c_i)$ for $i=1,k$.  By the definition of a
  homomorphism, $H$ contains both transitions
  $f(h_1(c_1),\ldots, h_1(c_k))\rightarrow h_1(c)$ and
  $f(h_2(c_1),\ldots, h_2(c_k))\rightarrow h_2(c)$, and we conclude
  $h_1(c)=h_2(c)$ because $H$ is deterministic.
\end{proof}

We call a tree automaton $\calA$ rigid~\cite{DBLP:journals/ejc/HellN91} if the identity mapping is the
only homomorphism $\calA\rightarrow \calA$.  It follows from
\autoref{lem:unique-morphism} that every E-graph is a rigid tree
automaton.


\begin{lem}\label{lem:egraph-partial-order}
  $\sqsubseteq$ over E-graphs forms a partial order up to isomorphism.
\end{lem}

\begin{proof}
  Obviously $\sqsubseteq$ is reflexive and transitive.  To prove
  anti-symmetry, assume two homomorphisms $h : G \rightarrow H$,
  $h' : H \rightarrow G$.  The composition $h' \circ h$ is a
  homomorphism $G \rightarrow G$, and, by uniqueness, it must be the
  identity on $G$; similarly, $h \circ h'$ is the identity on $H$,
  proving that $h$ is an isomorphism, thus $G, H$ are isomorphic.
%
\end{proof}

\revinsAll{Next, we define models for term rewriting systems.}

\begin{dfn}
  We say that an E-graph $H=\langle Q,\Sigma,\Delta\rangle$ is a
  \emph{model} of a TRS $\trs$ if, for every rule
  $\lhs{}\rightarrow\rhs{}$ in $\trs$ and any substitution
  $\sigma:\varset(\lhs)\rightarrow Q$, if
  $\lhs[\sigma] \rightarrow^*_G c$ then
  $\rhs[\sigma] \rightarrow^*_G c$.
%
  If $G$ is another E-graph, then we say that $H$ is a {\em model} for
  the pair $\trs, G$ if $G \sqsubseteq H$ and $H$ is a model of
  $\trs$.  $H$ is a {\em universal model} if for any other model $H'$,
  it holds that $H \sqsubseteq H'$.
\end{dfn}

When it exists, the universal model is unique up to isomorphism,
because $H\sqsubseteq H'$ and $H'\sqsubseteq H$ implies $H,H'$ are
isomorphic.  

\revinsA{Continuing \autoref{ex:intro:revisited}, let $\trs$
  consists of the rule $f(x,x) \rightarrow g(x,x)$, and let $G, H$ be
  the E-graphs in \autoref{fig:egraph}.  $G$ is not a model of $\trs$,
  because for the substitution $\sigma(x)=c_1$ we have
  $\lhs[\sigma]=f(c_1,c_1)\rightarrow_G c_2$, but
  $\rhs[\sigma]=g(c_1,c_1)\not\rightarrow^*_G c_2$.  On the other
  hand, one can check that $H$ is a model of $\trs$; in fact it is a
  model of $\trs, G$, because $G \sqsubseteq H$.

  Given an E-graph $G$ and a TRS $\trs$, equality saturation
  constructs a universal model $H$ of $\trs, G$, by repeatedly
  applying some simple operations on $G$, which we define next.}

\vspace{-1em}

\subsection{Operations over E-graphs}
\label{sec:egraph-operations}


E-matching, Insertion, and Rebuilding are the building blocks of
equality saturation.  They defined in the literature
operationally~\cite{efficient-e-matching}.  We provide here a formal
definition, using tree automata terminology.  \revinsA{Throughout this
  section we fix an E-graph $G = \langle Q,\Sigma, \Delta \rangle$}.

{\bf E-matching} a rule $\lhs{}\rightarrow \rhs{} \in \trs$ in $G$
returns the set of pairs $(\sigma, c)$, where
$\sigma : \varset(\lhs)\rightarrow Q$ is a substitution such that
$\lhs[\sigma] \rightarrow^*_G c$.  For example, considering the
E-graph on the left of \autoref{fig:egraph} and the rule
$f(x,x)\rightarrow g(x,x)$, E-matching returns
\revdelA{\((c_{i+1}, \{x\mapsto c_i,y\mapsto c_i\})\)}
\revinsA{\((\{x\mapsto c_i,y\mapsto c_i\},c_{i+1})\) } for
$i=1,2,3$. E-matching is analogous to computing the triggers for a TGD
or EGD (Sec.~\ref{sec:chase}).
%
%

The {\bf Insertion} of a pair $(t,c)$ into $G$, where
$t \in T(\Sigma\cup Q)$ and $c \in Q$, returns an automaton $\calA$
such that $G\sqsubseteq \calA$ such that $t \rightarrow^*_{\calA} c$.
To define $\calA$, we need the following:

\begin{dfn} Fix a term $t \in T(\Sigma \cup Q)$ and a state $c \in Q$.
  The {\em flattening for $t$ with root $c$}, in notation
  $\flatt(t \rightarrow^* c)$, or just $\flatt$ when $t, c$ are clear
  from the context, is an E-graph that has one distinct state $q_u$
  for each subterm $u$ of $t$, and has a transition
  $f(q_{u_1}, \ldots, q_{u_k}) \rightarrow q_u$, for all subterms $u$
  of the form $u = f(u_1, \ldots, u_k)$ and $f \in \Sigma$.  Moreover,
  it is enforced that the state of the root node is $c$ (i.e.,
  $q_t = c$).  One can check that $t \rightarrow^*_{\flatt} c$, and
  that $\approx_{\flatt}$ is the identity on all subterms of
  $t$. Flattening is also called \emph{normalization}~\cite{tac-termination}.
\end{dfn}
%

\begin{wrapfigure}{r}{0.6\textwidth}
  \vspace{-0.5cm}
  \centering
  \begin{subfigure}[b]{0.3\linewidth}
    \centerline{$G:$}
    \centerline{\includegraphics[height=3.8cm]{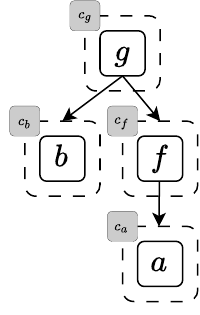}}
  \end{subfigure}
  \begin{subfigure}[b]{0.42\linewidth}
    \centerline{$\calA:$}
    \centerline{\includegraphics[height=3.8cm]{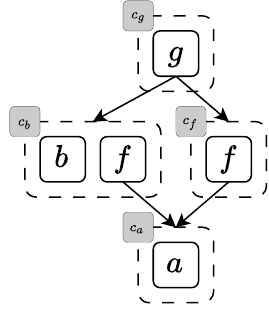}}
  \end{subfigure}
  \hfill
  \begin{subfigure}[b]{0.25\linewidth}
    \centerline{$H:$}
    \centerline{\includegraphics[height=3.8cm]{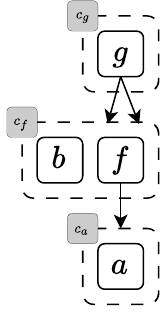}}
  \end{subfigure}
  \caption{Example E-graphs on insertion and rebuilding.}
  \vspace{-0.5cm}
  \label{fig:rebuilding}
\end{wrapfigure}
The result of {\em inserting} $(t,c)$ in $G$ is
$\calA \defeq G \cup \flatt(t\rightarrow^* c)$ (i.e.  we take the
set-union of all states and all transitions).  The result $\calA$ is a
reachable tree automaton, but it is non-deterministic in general, thus
it is not an E-graph; the next operation, rebuilding, converts it back
into an E-graph.  $G \sqsubseteq \calA$ holds, because the inclusion
$G\rightarrow\calA$ is a homomorphism.

As an example, if we insert the pair $(f(c_a), c_b)$ in the E-graph
$G$ in \autoref{fig:rebuilding}, the result is $\calA$ in the center
of the figure;
flattening $\flatt(f(c_a)\rightarrow c_b)$ has a single transition
$f(c_a)\rightarrow c_b$.


{\bf Rebuilding} converts $\calA$ into a deterministic automaton.
Formally, let $\calA$ be any reachable tree automaton, and recall that
$\calA$ may be infinite.  The \emph{congruence closure}
$\congr(\calA)$ is an E-graph (i.e. deterministic, reachable
automaton) such that $\calA \sqsubseteq \congr(\calA)$ and, for any
other E-graph $G'$, if $\calA \sqsubseteq G'$, then
$\congr(\calA)\sqsubseteq G'$.  \revinsC{We prove the following in \cref{apdx:rebuilding-exists}:}

\begin{restatable}{lem}{rebuildingexists}
  \label{lem:rebuilding-exists}
  For any reachable tree automaton $\calA$, $\congr(\calA)$
  exists and is unique.
\end{restatable}

The procedure of computing $\congr(\calA)$, also known as
\emph{rebuilding} in EqSat literature~\cite{egg}, can be done
efficiently in the finite case, for instance with Tarjan's algorithm
\cite{tarjan-congruence}.  The idea is to repeatedly find violating
transitions $f(c_1,\ldots, c_k)\rightarrow c$ and
$f(c_1,\ldots, c_k)\rightarrow c'$ with $c\neq c'$, and replace every
occurrence of $c$ with $c'$, until fixpoint.  \revinsC{This is similar
  to determinizing $\calA$, but instead of constructing powerset
  states like $\set{c_1, c_4, c_7}$, we equate states $c_1=c_4=c_7$;
  thus, $\congr(\calA)$ has at most as many states as $\calA$, and the
  procedure always terminates for finite $\calA$, as merging shrinks
  the number of states.}

\begin{exm}
  The tree automaton $\calA$ in \autoref{fig:rebuilding} is
  non-deterministic, because $f(a) \rightarrow^* c_b$ and
  $f(a) \rightarrow^* c_f$.  The congruence closure algorithm merges
  $c_b$ and $c_{f}$, and produces the E-graph $H$ in
  \autoref{fig:rebuilding}.  Notice that $H$ represents strictly more
  terms than $\calA$.  For example, $H$ represents $g(b,b)$, because
  $g(b,b) \rightarrow^*_H c_g$, but $\calA$ does not represent
  $g(b,b)$.
\end{exm}

\mysubparagraph{Least upper bound of E-graphs}

Let $(G_i)_{i \in I}$ be a (possibly infinite) family of E-graphs.
Recall that their least upper bound $G$ is an E-graph such that $G$ is
an upper bound for every E-graph in the set, $G_i \sqsubseteq G$ for
all $i \in I$, and for any other upper bound $G'$, it holds that
$G \sqsubseteq G'$.  We prove the following in
\cref{apdx:least-upper-bound}:

\begin{restatable}[Least upper bound]{lem}{leastupperbound}
  \label{lemma:union}
  The least upper bound exists and is given by $\congr \left( \mathcal{A}\right)$,
  where $\calA$ is the automaton consisting of the disjoint union of
  the states and the disjoint union of the transitions of all E-graphs
  $G_i$.
\end{restatable}

We will denote the least upper bound by $\bigsqcup_{i \in I} G_i$. It
is also possible to show that every family of E-graphs admits a
greatest lower bound, by using a product
construction~\cite{gulwani2005join}, but we do not need it in this
paper.

%

    
\subsection{Equality saturation}

\label{subsec:eqsat}

\revinsA{ The standard definition of equality saturation in the
  literature is procedural: given an E-graph
  $G=\langle Q,\Sigma,\Delta\rangle$ and TRS $\trs$, equality
  saturation repeatedly applies matching/insertion/rebuilding.  EqSat
  is undefined when this process does not terminate.  We provide here
  an alternative definition, as the least fixpoint of an {\em immediate
  consequence operator} (ICO), and prove that it always exists.  We start
  by introducing the ICO:}
\begin{align}
&&  \ico_{\trs} \defeq\  & \congr \circ T_{\trs}  \label{eq:def:ico}
\end{align}
$T_{\trs}$ is the match/apply operator: it computes all E-matches then
inserts the $\rhs$'s into $G$:
\begin{align*}
  T_{\trs}(G) \defeq G \cup \bigcup \setof{\flatt(\rhs[\sigma]  \rightarrow^*_G c)}{(\lhs\rightarrow\rhs) \in \trs, \sigma:\varset(\lhs)\rightarrow Q, \lhs[\sigma] \rightarrow^*_G c}
\end{align*}
$\congr$ is the rebuilding operator
of~\autoref{lem:rebuilding-exists}.  \revdelA{We show:}

\begin{lem} \label{lemma:ico:monotone} $\ico_{\trs}$ is inflationary
  ($G \sqsubseteq \ico_{\trs}(G)$ for all $G$) and monotone.
\end{lem}

\begin{proof} \revinsC{That $\ico_{\trs}$ is inflationary}\revdelC{Inflationary} follows from
  $G\sqsubseteq T_{\trs}(G) \sqsubseteq \congr(T_{\trs}(G))$.  We
  prove that both $T_{\trs}$ and $\congr$ are monotone.  Let
  $H\defeq \bigcup \{\flatt(\rhs[\sigma] \rightarrow^*
    c)\mid (\lhs\rightarrow\rhs) \in \trs, \sigma:\varset(\lhs)\rightarrow Q,
    \lhs[\sigma] \rightarrow^*_G c\}$, thus $T_{\trs}(G) = G \cup H$.
  Any homomorphism $G \rightarrow G'$ can be extended to a
  homomorphism $G \cup H \rightarrow G' \cup H$, which proves that
  $T_{\trs}$ is monotone.
  Consider two automata $\calA,\calA'$ and assume
  $\calA \sqsubseteq \calA'$, i.e. there exists a homomorphism from
  $\calA$ to $\calA'$.  Denote $G \defeq \congr(\calA)$,
  $G'\defeq \congr(\calA')$.  Then,
  $\calA \sqsubseteq \calA' \sqsubseteq G'$, which implies
  $\calA \sqsubseteq G'$.  By the definition of $G = \congr(\calA)$,
  we have $G \sqsubseteq G'$, proving that $\congr$ is monotone.
\end{proof}

With Lemmas~\ref{lemma:union} and~\ref{lemma:ico:monotone},
 we show the following in \cref{apdx:eqsatdfn}:

\begin{restatable}{thm}{eqsatdfnthm} 
  \label{thm:eqsat:dfn}
  Fix an E-graph $G$, and consider the class $\calC_G$ of
  E-graphs $G'\sqsupseteq G$. Then
  $\ico_{\trs}: \calC_G \rightarrow \calC_G$ has a least fixpoint,
  given by
  \begin{align}
&&    \eqsat(\trs,G) \defeq& \bigsqcup_{i\geq 0} \ico_{\trs}^{(i)}(G) \label{eq:eqsat:def}
  \end{align}
  Furthermore, $\eqsat(\trs, G)$ is a universal model of $\trs,G$; we
  call it {\em equality saturation}.
\end{restatable}

%

\revinsA{Given $G, \trs$, our semantics of EqSat is the least fixpoint
  in~\eqref{eq:eqsat:def}, which is also the unique universal model of
  $G, \trs$.  When $\eqsat(\trs,G)$ is finite, then this coincides
  with the standard procedural definition in the literature.}  A
common case (e.g., in program and query optimization settings) is when
$G$ represents a single term $t$, more precisely
$G = \flatt(t \rightarrow^* c)$ with fresh state $c$; in that case we
denote $\eqsat(\trs, G)$ as $\eqsat(\trs, t)$.

\mysubparagraph{Properties of equality saturation}

We establish several basic facts of $\eqsat$.

\begin{lem}[Inflationary]~\label{lemma:inflationary}
  $G \sqsubseteq \eqsat(\trs, G)$.
\end{lem}


\begin{cor} \label{cor:calg:calrg}
$\calL(G) \subseteq \calL(\eqsat(\trs,G))$ and $(\approx_G) \subseteq (\approx_{\eqsat(\trs,G)})$.
\end{cor}

\autoref{lemma:inflationary} follows
from~\autoref{lemma:ico:monotone}, while~\autoref{cor:calg:calrg}
follows from~\autoref{lemma:homomoprhism:simple}
and~\autoref{lemma:approx:g:h}.  Thus, $\eqsat(\trs,G)$ represents
more terms than $G$, and identifies more pairs of terms than $G$.
Next, we examine the relationship between the PCR defined by
$\eqsat(\trs,G)$ and the relations $\trs^*$ and $\approx_{\trs}$
defined by the TRS $\trs$ (see \autoref{eq:trs}).  If
$\approx_1, \approx_2$ are two PCRs, then we denote by
$\approx_1 \vee \approx_2$ the smallest PCR that contains both.  We
prove:

\begin{lem}[Representation]
  \label{lemma:representation}
  Let $w \in T(\Sigma)$ be a term represented by some state of the
  E-graph $H \defeq \eqsat(\trs, G)$.  The following hold:
  $\trs^*(w) \subseteq [w]_{\approx_H} \subseteq [w]_{\approx_{\trs}\vee\approx_G}$.
\end{lem}

\begin{proof}
  The definitions of E-matching and insertion imply that, if
  $u \rightarrow_\trs v$ and $u$ is represented by some state $c$ of
  some E-graph $K$, then $v$ is represented by the same state of the
  E-graph $\ico(K) = \congr(T_{\trs}(K))$.  Therefore, if
  $u \rightarrow_\trs v$ and $u$ is represented by some state of $H$,
  then $v$ is represented by the same state of $\ico(H) = H$ (because
  $H$ is a fixpoint of $\ico$).  This implies that
  $\trs^*(w) \subseteq [w]_{\approx_H}$.

  For the second part, we denote by $G_k = \ico^{(k)}(G)$, and check
  by induction on $k$ that $[w]_{\approx_{G_k}} \subseteq
  [w]_{\approx_{\trs}\vee\approx_G}$.  When $k=0$ then $G_0 = G$ and
  the claim is obvious.  For the inductive step we observe that the
  only new identities introduced by $G_{k+1}$ are justified by $\trs$.
%
\end{proof}


\revinsA{In other words, if $w \rightarrow^*_{\trs} v$, then
  $\eqsat(\trs,G)$ will equate $w$ with $v$; and if $\eqsat(\trs,G)$
  equates $w$ with $v$ then this can be derived from $\approx_{\trs}$
  and $\approx_G$.  In general, $w \approx_{\trs} v$ does not imply
  $w \approx_H v$. For a simple example, let
  $\trs = \set{a \rightarrow b}$, thus $a \approx_{\trs} b$, and let
  $G$ represent only the term $b$.  Then $H = \eqsat(\trs, G)= G$
  represents only the term $b$, thus $a \not\approx_H b$.}

\begin{exm} \revinsA{For some TRS, the starting E-graph $G$ determines
    whether EqSat terminates in a finite number of steps.  For a
    simple example, consider $\Sigma=\{f(\cdot),g(\cdot),a\}$,
    $\trs=\{f(g(x))\rightarrow g(f(x)) \}$.  If the initial E-graph $G$
    represents only the term $f(g(a))$ (and its subterms), 
    then EqSat terminates, and the
    resulting $H \defeq \eqsat(\trs, G)$ represents a PCR where
    $f(g(a)) \approx_H g(f(a))$.  On the other hand, if $G$ is the
    E-graph with states $c_f, c_g$ and transitions
    $\{g(c_f)\rightarrow c_g, f(c_g)\rightarrow c_f, a\rightarrow
    c_f\}$, then $G$ already represents infinitely many terms
    $\calL(c_f)\cup\calL(c_g)$, where
    $\calL(c_f)=\set{a, f(g(a)), f(g(f(g(a)))), \ldots}$ and
    $\calL(c_g)=\set{g(a), g(f(g(a))), \ldots}$.  After equality
    saturation, $H = \eqsat(\trs,G)$ represents all terms in
    $T(\Sigma)$ where the numbers of occurrences $\#f,\#g$ of
    $f,g$ satisfy $\#f \leq \#g \leq \#f+1$.   This is not a regular
    language, hence $H$ is infinite, and EqSat will not terminate in a
    finite number of steps.
}
\end{exm}


\begin{exm} We show that both inclusions
  in~\autoref{lemma:representation} can be strict.  Let
  $\trs= \{a\rightarrow b, c\rightarrow b\}$ and let $G$ be the
  E-graph representing a single term $f(a, b)$ with transitions
  \( \{\ a\rightarrow c_a, b\rightarrow c_b, f(c_a, c_b)\rightarrow
  c_f\ \} \).  Then $H = \eqsat(\trs,G)$ has transitions:
  \( \{\ a\rightarrow c_a, b\rightarrow \underline{c_a}, f(c_a,
  \underline{c_a})\rightarrow c_f\} \).  We have:\footnote{
    \revinsC{We use $f(a|b, b)$ as a shorthand for
      $\{f(t, b)\mid t=a\lor t=b\}$ (as in regular expressions) and
      similarly for other terms.}}  $\trs^*(f(a,b))=f(a|b, b)$,
  $[f(a,b)]_{\approx_H}=f(a|b, a|b)$, and
  $[f(a,b)]_{\approx_{\trs}\vee \approx_G}=[f(a,b)]_{\approx_{trs}}=
  f(a|b|c, a|b|c)$. All three sets are different.  
\end{exm}

However, the three expressions in ~\autoref{lemma:representation} are
equal in an important special case:

\newcommand{\Sym}{\textit{Sym}}

\begin{cor}\label{cor:var-preserving}
  Suppose $\trs$ is a variable-preserving term rewriting system.  Let
  $\textit{Sym}\left(\trs\right) = \trs\cup\trs^{-1}$, and let
  $w\in T(\Sigma)$ be a term represented by some state of the E-graph
  $H^{\leftrightarrow}\defeq \eqsat(\Sym(\trs), G)$.  The following
  hold:
  $\left(\Sym\left(\trs \right)\right)^*(w)=
  [w]_{\approx_{H^\leftrightarrow}}= [w]_{\approx_\trs\vee
    \approx_G}$.
\end{cor}

\revinsA{Let $G$ be a finite E-graph. If $\eqsat(\trs,G)$ is infinite, then 
   EqSat does not terminate in a finite number of
   steps.  Somewhat surprisingly, the converse does hold: if
   $\eqsat(\trs,G)$ is finite, then EqSat terminates in
   a finite number of steps.  This follows from the next lemma, whose
   proof is deferred to~\cref{apdx:convergence}.}

\begin{restatable}[Finite convergence]{lem}{convergence}
  \label{lemma:convergence}
  Let $\mathcal{G}:G_1\sqsubset G_2\sqsubset \ldots$ be an ascending sequence of finite E-graphs.
  If $G_{\infty}= \bigsqcup_i G_i$ is finite,
  then the sequence $\mathcal{G}$ is finite.
\end{restatable}

\begin{restatable}[Finite convergence of EqSat]{cor}{cor-convergence}
  \label{cor:convergence}
  Let $\trs$ be a term rewriting system and $G$ be a finite E-graph.
  If $\eqsat(\trs, G)$ is finite,
  EqSat converges in a finite number of steps.
\end{restatable}

\section{Equality saturation and the chase}

\label{sec:eqsat:and:chase}
\revinsA{
In this section, we briefly review necessary background on databases and the chase.
Then, we will show the fundamental connections between equality saturation and the chase.

\subsection{The Chase Procedure}\label{sec:chase}

\mysubparagraph{Databases and conjunctive queries}

A relational database \emph{schema} is a tuple of relation names
$\mathcal{S}=(R_1,\ldots, R_m)$ with associated arities $\arity(R_i)$.
A database instance is a tuple of relation instances
$I=(R_1^I,\ldots, R_m^I)$, where $R_i^I\subseteq \dom^{\arity(R_i)}$
for some domain $\dom$.  We allow an instance to be
infinite.  We often view a tuple $\vec{a}$ in $R_i^I$ as an atom
$R_i(\vec{a})$, and view the instance $I$ as a set of atoms.  The
domain $\dom$ is the disjoint union of set of \emph{constants} and a
set of \emph{marked nulls}.

A \emph{conjunctive query} $\lambda(\vec{x})$ is a formula with free
variables $\vec{x}$ of the form
$R_1(\vec{x_1})\land\ldots\land R_k(\vec{x_k})$, where each
$\vec{x}_i$ is a tuple of variables from $\vec{x}$.  The
\emph{canonical database} of a conjunctive query consists of all the
tuples $R_i(\vec{x_i})$, where the variables $\vec{x}$ are considered
marked nulls.

Let $I$, $J$ be two database instances. A \emph{homomorphism} from $I$
to $J$, in notation $h:I \rightarrow J$, is a a function
$h:\dom(I)\rightarrow\dom(J)$ that is the identity on the set of
constants, and maps each atom $R(\vec{a})\in I$ to an atom
$R(h(\vec{a}))\in J$. The notion of homomorphism immediately
extends to conjunctive queries and/or database instances.  The output
of a conjunctive query $\lambda(\vec{x})$ on a database $I$ is defined
as the set of homomorphisms from $\lambda(\vec{x})$ to $I$.  We say
that a database instance $I$ \emph{satisfies} a conjunctive query
$\lambda(\vec{x})$, denoted by
$I\models \exists \vec{x}\lambda(\vec{x})$, if there exists a
homomorphism $\lambda(\vec{x}) \rightarrow I$.

\mysubparagraph{Dependencies}

TGDs and EGDs describe semantic constraints between relations.
A TGD is a first-order formula of the form \(\lambda(\vec{x},\vec{y})\rightarrow \exists \vec{z}. \rho(\vec{x},\vec{z})\)
where $\lambda(\vec{x}, \vec{z})$ and $\rho(\vec{x}, \vec{y})$ are conjunctive queries with free variables
in $\vec{x}\cup\vec{y}$ and $\vec{x}\cup\vec{z}$.
An EGD is a first-order formula of the form \(\lambda(\vec{x})\rightarrow x_i=x_j\)
where $\lambda(\vec{x})$ is a conjunctive query with free variables in $\vec{x}$ and $\{x_i, x_j\}\subseteq \vec{x}$.


Fix a set of TGDs and EGDs $\Gamma$.  If $I$ is a database instance
and $d \in \Gamma$, then a \emph{trigger} for $d$ in $I$ is a
homomorphism from $\lambda(\vec{x}, \vec{y})$ (resp.\
$\lambda(\vec x)$) to $I$.  An \emph{active trigger} is a trigger $h$
such that, if $d$ is a TGD, then no extension $h$ to a homomorphism
$h' :\rho(\vec{x},\vec{z}) \rightarrow I$ exists, and, if $d$ is an
EGD, then $h(x_i)\neq h(x_j)$.  We say that $I$ is model for $\Gamma$,
and write $I\models \Gamma$, if it has no active triggers.

Given $\Gamma$ and $I$ we say that some database instance $J$ is a
{\em model} for $\Gamma, I$, if $J\models \Gamma$ and there exists a
homomorphism $I \rightarrow J$.  $J$ is called \emph{universal model}
if there is a homomorphism from $J$ to every model of $\Gamma$ and
$I$.  Universal models are unique up to homomorphisms.  

\mysubparagraph{The chase}

The chase is a fixpoint algorithm for computing universal models. 
We consider two variants of the chase here: the
\emph{standard chase} and \emph{Skolem chase}.
Both the standard chase and the Skolem chase produce a universal model
of $\Gamma, I$~\cite{FAGIN200589,bench-chase,marnette2009generalized}.
The standard chase computes answers by deriving a sequence of
\emph{chase steps} until all dependencies are satisfied.  A chase
step, denoted as $I\xrightarrow[]{d, h} J$, takes as inputs an
instance $I$, a homomorphism $h$, and a dependency $d$, where $h$
is an active trigger of $d$ in $I$, and produces an output instance
$J$ by adding some tuples (for TGDs) or collapsing some elements (for
EGDs).  Specifically, if $d$ is a TGD, the chase step extends $I$ with
the tuple $h'(\rho(\vec{x},\vec{z}))$, where $h'$ is an extension of
$h$ that maps the variables $\vec{z}$ on which $h$ is undefined to
fresh marked nulls.  If $d$ is an EGD, if $h(x_i)$ (or $h(x_j)$) is a
marked null, a chase step replaces in $I$ every occurrence of $h(x_i)$
with $h(x_j)$ (or $h(x_j)$ with $h(x_i)$).  If neither $h(x_i)$ nor
$h(x_j)$ is a marked null and $h(x_i)\neq h(x_j)$, then the chase
fails.  


A standard chase sequence starting at $I_0$ is a sequence of
successful chase steps
$I_0\xrightarrow{d_1,h_1}I_1 \xrightarrow{d_2,h_2}\cdots$ that is
\emph{fair}: for all $i\geq 0$, for each dependency $d$ and active
trigger $h$ of $d$ in $I_i$, some $j\geq i$ must exist such that $h$
is no longer an active trigger of $d$ in $I_j$.  The \emph{result} of
a (possibly infinite) chase sequence is
$\bigcup_{i\geq 0} \bigcap_{j\geq i} I_j$~\cite{bench-chase}.  
A chase sequence is
\emph{terminating} if it ends with $I_n$ and $I_n\models \Gamma$, in
which case $I_n$ is the result of the chase sequence.  The standard
chase is non-deterministic: depending on the order of firing, the
chase sequence can be different.  Different chase sequences can even
differ on whether they terminate.

The Skolem chase, discussed in~\cite{marnette2009generalized}, differs
from the standard chase in several ways.  It first \emph{skolemizes}
each TGD
$d: \lambda(\vec{x},\vec{y})\rightarrow \exists
\vec{z}. \rho(\vec{x},z_1,\ldots, z_k)$ to
\( \lambda(\vec{x},\vec{y})\rightarrow
\rho(\vec{x},f^d_{z_1}(\vec{x}), \ldots, f^d_{z_k}(\vec{x})), \) where
each $f^d_{z_j}$ is an uninterpreted function from $\dom^{|\vec{x}|}$
to $\dom$.  The result of the Skolem chase, denoted as $\sklchase(\Gamma, I)$, is the least fixpoint of
the immediate consequence operator (ICO) of the Skolemized TGDs.  
Note that the Skolem chase does not directly handle EGDs but uses a technique
 called \emph{singularization} \cite{marnette2009generalized}
 to simulate EGDs with TGDs.
}

\subsection{Reducing the Skolem chase to equality saturation}
\label{sec:chase-to-eqsat}

In this section, we show how to reduce
 the Skolem chase to EqSat.
We only consider TGDs,
 since in the Skolem chase,
 EGDs are modeled as TGDs
 using singularization \cite{marnette2009generalized}.

We show an encoding 
 where
 there exists a simple mapping
 from E-graphs to database instances, defined by
\[\revinsA{\xi(G) = \{R(c_1,\ldots, c_k)\in\mathcal{L}(G)\mid R \text{ is a relation symbol in } \mathcal{S}\}}\]
 such that,
 given a set of dependencies $\Gamma$, 
 running EqSat on an encoded term rewriting system from $\Gamma$
 corresponds to running the Skolem chase on the set of dependencies via $\xi$.
\revinsA{Intuitively, given an E-grpah, 
$\xi$ collects every term that corresponds to a tuple
from the language of $G$.}
An illustration of $\xi$ is shown in \autoref{fig:skolem-to-eqsat}.

\begin{restatable}{thm}{skolemtoeqsatthm}
    \label{thm:skolem-to-eqsat}
Given a database schema $\mathcal{S}=(R_1, \ldots, R_m)$, a set of TGDs $\Gamma$,
 and an initial database $I$,
 it is possible to define 
 a signature $\Sigma$, a term rewriting system $\trs$ over $\Sigma$,
 and an initial term $t$ such that
\[ \xi(\eqsat(\trs, t)) = \sklchase(\Gamma, I).\]
Moreover, the Skolem chase terminates 
 if and only if equality saturation terminates.
\end{restatable}

The intuition for the construction is that we can uniformly 
 treat relational atoms as E-nodes contained in a special E-class,
 and Skolem functions naturally correspond to terms in EqSat.
More specifically,

\newcommand{\boldland}{\boldsymbol{\land}}
\newcommand{\boldtop}{\boldsymbol{\top}}

\begin{itemize}
    \item Add symbols $\{\boldtop, \boldland(\cdot, \cdot)\}$ to the signature $\Sigma$.
    Add rewrite $r_{\boldtop}:\boldtop \rightarrow \boldland(\boldtop, \boldtop)$ to $\trs$.
    Let the initial term $t$ be $\boldtop$.
    \item Add every Skolem function symbol to $\Sigma$, and for every $n$-ary relational symbol $R\in \mathcal{S}$, add a $n$-ary function symbol 
    to $\Sigma$, and add rewrite rule $r_R:R(x_1,\ldots, x_n)\rightarrow \boldtop$.

    \item For every Skolemized TGD
    \[
        d: \ R_1\left(\vec{x_1},\vec{y_1}\right)\land\ldots \land R_n\left(\vec{x_n}, \vec{y_n}\right)
        \rightarrow 
        R'_1\left(\vec{x'_1}, \vec{f^d_z}_1\right)\land \ldots\land R'_m\left(\vec{x'_m}, \vec{f^d_z}_m\right),
    \]
    replace the conjunctions in the head and body with nested applications of
    $\boldland$ and $\boldtop$:
   \[r_d:\ \boldland\left(R_1\left(\vec{x_1},\vec{y_1}\right),
    \boldland\left(\ldots 
    \boldland\left( R_n\left(\vec{x_n}, \vec{y_n}\right), \boldtop\right)\right)
    \right)
    \rightarrow 
    \boldland\left(
    R'_1\left(\vec{x'_1}, \vec{f^d_z}_1\right),
    \boldland\left(\ldots
    \boldland\left(
     R'_m\left(\vec{x'_m}, \vec{f^d_z}_m\right)
     , \boldtop\right)\right)\right)\]
    and add $r_d$ to $\trs$.
    
    \item For each constant $c$ in the input database $I$,
    add a nullary function symbol $c$ to $\Sigma$.
    For each tuple $t=R(c_1, \ldots, c_n)$ in the input database $I$,
    add rewrite  $r_t:\boldtop \rightarrow R(c_1, \ldots, c_n)$.
\end{itemize}

\begin{figure}
    \centering
    \includegraphics[width=0.8\textwidth]{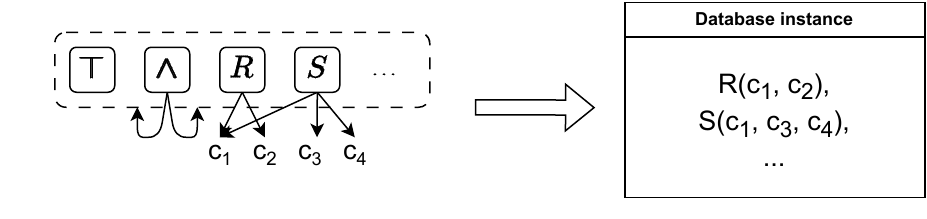}
    \caption{Mapping results of encoded EqSat back to database instances.}
    \label{fig:skolem-to-eqsat}
\end{figure}

\begin{proof}[Proof of \autoref{thm:skolem-to-eqsat}]
    See \cref{apdx:skolem-to-eqsat}.
\end{proof}


\subsection{Reducing equality saturation to the standard chase}
\label{sec:eqsat-to-chase}

We show how to reduce
 equality saturation to the standard chase.
The encoding itself is straightforward.
However, the standard chase is non-deterministic and can have different chase sequences, 
 so a natural question is what kind of the chase sequence will converge finitely,
 given that EqSat terminates, and vice versa.
We show that as long as the chase sequence applies EGDs frequent enough,
 the chase sequence will always converge.
We capture this notion as EGD-fairness.

\begin{dfn}\label{dfn:egd-fairness}
    Given a database schema $\mathcal{S}$, a set of dependencies $\Gamma$ over $\mathcal{S}$,
    and an initial database $I_0$. 
    We call a chase sequence $I_0,I_1,\ldots$ of $\Gamma$ and $I$ EGD-fair if
    for every $i$, either $I_i$ is a model of $\Gamma$ and the chase terminates,
    or there exists some $j>i$ such that $I_j$ is a model of the EGD subset of $\Gamma$.
\end{dfn}
Given \revinsA{that} EqSat terminates,
 what can we say about chase sequences that are not EGD-fair?
 In fact, such chase sequences may not terminate.
 Despite this,
 it can be shown that the result of such chase sequences, terminating or not,
 is isomorphic to the result of equality saturation (when encoded as a database).
On the other hand,
 to show \revinsC{that }equality saturation terminates,
 it is sufficient to show an arbitrary chase sequence terminates.

The following theorem shows the connection between EqSat and the standard chase.

\begin{restatable}{thm}{eqsattochasethm}
    \label{thm:eqsat-to-chase}
    Given signature $\Sigma$, a set of rewrite rules $\trs$ over $\Sigma$,
    and an initial E-graph $G$,
    it is possible to define 
    a relational schema $\mathcal{S}$, a set of dependencies $\Gamma$ over $\mathcal{S}$,
    and an initial database $I$ over $\mathcal{S}$.
    The following three statements are equivalent:
    \begin{enumerate}
        \item Equality saturation terminates for $\trs$ and $t$.
        \item There exists a terminating chase sequence of the standard chase for $\Gamma$ and $I$.
        \item All EGD-fair chase sequences of the standard chase terminate for $\Gamma$ and $I$.
    \end{enumerate}
    Moreover, denote the result of an arbitrary chase sequence as $I_\infty$.
    If equality saturation terminates,
    $I_{\infty}$ is isomorphic to the database encoding the resulting E-graph of $\eqsat(\trs, G)$.
\end{restatable}

The encoding consists of two steps. 
First, we can encode an E-graph as a database.
Second, we encode the match/apply operator and congruence closure operator 
 as a set of TGDs and EGDs.
To encode an E-graph as a database:
 \begin{itemize}
    \item Take the domain $\dom$ to be the set of all E-classes, which are treated as marked nulls.
    \item For every function symbol $f$ of arity $n$,
    add relation symbol $R_f$ of arity $n+1$ to $\mathcal{S}$.
    \item For every E-node $f(c_1, \ldots, c_n)\rightarrow c$,
    add a tuple $R_f(c_1, \ldots, c_n, c)$ to the database $I$.
 \end{itemize}
Under this encoding, each E-class is treated as a marked null, and
 each E-node is treated as a tuple.

The encoding of the match/apply operator and congruence closure operator is plain:
\begin{itemize}
    \item For every function symbol $f$ of arity $n$,
    add a functional dependency 
    $R_f(x_1,\ldots, x_n, x)\land R_f(x_1, \ldots, x_n, x')\rightarrow x=x'$ 
    to $\Gamma$.
    \item For every rewrite rule $\lhs\rightarrow \rhs$ in $\trs$,
    flatten the left- and right-hand side into conjunctions of relational atoms,
    unify the variable denoting the root node of $\lhs$ with that of $\rhs$,
    and add existential quantifiers to the head accordingly.
    For example, rule $f(f(x, y), z)\rightarrow f(x, f(y, z))$ is 
    flattened into $R_f(x, y, w_1)\land R_f(w_1, z, r)\rightarrow \exists w_2, R_f(x, y, w_2)\land R_f(w_2, v, r)$.

    There are two corner cases to the above translations. 
    First, if $\lhs$ is a single variable $x$,
    we need to introduce $n$ rules of the form $R_f(y_1, \ldots, y_k, x)\rightarrow \ldots$, 
    one for each function symbol, 
    to ``ground'' $x$.
    For instance, suppose $\Sigma=\{f(\cdot,\cdot), g(\cdot)\}$,
    rewrite rule $x\rightarrow g(x)$ is flattened into
    two dependencies: $R_f(y_1, y_2, x) \rightarrow R_g(x, x)$ 
    and $R_g(y_1, x)\rightarrow R_g(x, x)$.
    Second, in the case that the right-hand side is a single variable $x$,
    we need to add an EGD instead of a TGD.
    For example, 
    rule $f(x, y)\rightarrow x$ is encoded as an EGD $R_f(x, y, r)\rightarrow x=r$.
\end{itemize}

\begin{proof}[Proof of \autoref{thm:eqsat-to-chase}]
    See \cref{apdx:eqsat-to-chase}.
\end{proof}

%


\def\configlang{\emph{CONFIG}}

\section{The termination theorems of equality saturation}
\label{sec:terminations}

Finally, we present our main results here.

\begin{restatable}[Single-instance termination]{thm}{terminationthm}
    \label{thm:termination}
    The following problem is R.E.-complete:
    \begin{itemize}
        \item Instance: A term rewriting system \(R\), a term \(t\).
        \item Question: Does EqSat terminate with \(R\) and \(t\)?
    \end{itemize}
\end{restatable}

\begin{restatable}[All-term-instance termination]{thm}{allinstancethm}
  \label{thm:all-term-instance}
    The following problem is $\Pi_2$-complete:
    \begin{itemize}
        \item Instance: A term rewriting system \(R\).
        \item Question: Does EqSat terminate with \(R\) and \(t\) for all terms \(t\)?
    \end{itemize}
\end{restatable}

While \autoref{thm:termination} follows immediately from the fact the the Skolem chase is undecidable,
 our proof in Appendix~\ref{apdx:termination} 
 is based on Narendran et al.~\cite{narendran1985complexity},
 which allows us to also show \autoref{thm:all-term-instance}.
We encode a Turing machine as a term rewriting system with the property that
 the congruence classes of initial configurations corresponds to traces of running such configurations,
 and that EqSat terminates if and only if congruence class is finite.
For the all-term-instance case, 
 we then show that the congruence class of an arbitrary term is infinite 
  if and only if the congruence class of an initial configuration is.
The actual proof is slightly more involved so we refer the reader to \cref{apdx:termination} 
for more details.

The technique above does not apply to the all-E-graph-instance case,
 however. 
 The all-E-graph-instance termination can be thought of as 
 having inputs both a term and a set of ground identities,
 and we have no control over the latter.
Still, we are able to prove that this problem is undecidable
 by a reduction from the Post correspondence problem (\cref{apdx:all-egraph-instance}),
 while the exact upper bound is unknown.

\begin{restatable}[All-E-graph-instance termination]{thm}{allegraphinstancethm}
    \label{thm:all-egraph-instance}
    The following problem is undecidable:
    \begin{itemize}
        \item Instance: A term rewriting system \(R\).
        \item Question: Does EqSat terminate with \(R\) and \(G\) for all E-graphs \(G\)?
    \end{itemize}
\end{restatable}


\vspace{-0.5em}
\section{Weak term acyclity for equality saturation termination}

\label{sec:acyclicity}

\begin{wrapfigure}{r}{0.4\textwidth}
    \vspace{-1.5em}
    \centering
    \begin{subfigure}{0.4\textwidth}
        \centering
        \includegraphics[height=2cm]{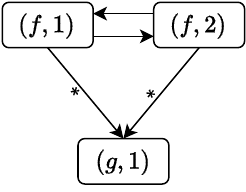}
        \caption{\autoref{ex:acyclicity-ex1}}
        \label{fig:acyclicity-ex1}
    \end{subfigure}
    \begin{subfigure}{0.4\textwidth}
        \centering
        \includegraphics[height=2cm]{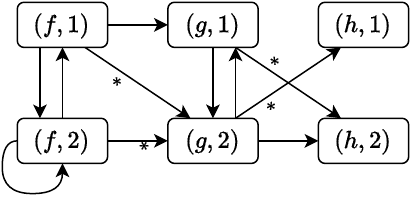}
        \caption{\autoref{ex:acyclicity-ex2}}
        \label{fig:acyclicity-ex2}
    \end{subfigure}
    \caption{
        Example weak term dependency graphs.
    Special edges are marked with $\ast$.}
    \label{fig:acyclicity-examples}
    \vspace{-2.5em}
 \end{wrapfigure}
We can adapt the classic weak acyclicity criterion \cite{FAGIN200589},
 which is used to show the termination of the chase algorithm,
 to equality saturation.
The adapted criterion, which we call \emph{weak term acyclicity},
 is more powerful
 than simply translating EqSat rules to 
 TGDs/EGDs and applying weak acyclicity.
We demonstrate weak term acyclicity with two examples,
 and the full definition can be found at \autoref{apdx:acyclicity}.

\begin{exm}\label{ex:acyclicity-ex1}
   Consider $\trs=\{f(f(x,y), z)\rightarrow g(f(z, x))\}$.
   This ruleset is weakly term acyclic, with the weak term dependency graph shown in \autoref{fig:acyclicity-ex1}.
   Note however if we derive the dependencies $\Gamma$ using the method in \autoref{sec:eqsat-to-chase} from $\trs$,
    $\Gamma$ is not weakly acyclic.
\end{exm}

\begin{exm}\label{ex:acyclicity-ex2}
   Consider \(\trs=\{
       g(f(x_1, y_1), f(z_1, x_1))\allowbreak\rightarrow g(z_1, f(y_1, x_1)),
       g(x_2, y_2)\rightarrow h(y_2, g(y_2, x_2))
   \}\).
   This ruleset is weakly term acyclic. Its weak term dependency graph is shown in \autoref{fig:acyclicity-ex2}.
\end{exm}

\vspace{-1em}
\revinsAll{\section{Conclusion}}
\revinsAll{
We have presented a semantic foundation for E-graphs and EqSat:
We identified E-graphs as reachable and deterministic tree automata 
and defined the result of EqSat as the least fixpoint according to E-graph homomoprhisms.
We defined the universal model of E-graphs and showed 
the fixpoint EqSat produces is the universal model (\autoref{thm:eqsat:dfn}).
We showed several basic properties about E-graphs, 
including a finite convergence lemma (\autoref{lemma:convergence}).
We then established connections between EqSat and the chase in both directions (\autoref{sec:eqsat:and:chase})
and characterize chase sequences that correspond to EqSat with EGD-fairness (\autoref{dfn:egd-fairness}).
Our main results are on the terminations of EqSat in three cases: single-instance, all-term-instance, and all-E-graph-instance.
Finally, adapting ideas from weak acyclicity for the chase,
we defined weak term acyclicity which implies EqSat termination.

The correspondence 
between EqSat and the chase established in this paper may help
further port the rich results of database theory to EqSat,
as the current paper only scratches the surface of the deep literatures of the chase.
Another direction is to use 
our better understanding of EqSat to design more efficient and expressive EqSat tools
and better support downstream applications 
of EqSat.
Finally, many problems about EqSat are still open.
For example, the exact upper bound of the all-E-graph-instance termination is not known.
Other problems include 
rule scheduling, evaluation algorithm, and E-graph extraction.
}

\bibliography{icdt-submission}

\begin{thebibliography}{10}

\bibitem{traat}
Franz Baader and Tobias Nipkow.
\newblock {\em Term Rewriting and All That}.
\newblock Cambridge University Press, USA, 1999.

\bibitem{abstract-congruence-closure}
Leo Bachmair, Ashish Tiwari, and Laurent Vigneron.
\newblock Abstract congruence closure.
\newblock {\em J. Autom. Reason.}, 31(2):129--168, 2003.

\bibitem{bench-chase}
Michael Benedikt, George Konstantinidis, Giansalvatore Mecca, Boris Motik,
  Paolo Papotti, Donatello Santoro, and Efthymia Tsamoura.
\newblock Benchmarking the chase.
\newblock In {\em Proceedings of the 36th ACM SIGMOD-SIGACT-SIGAI Symposium on
  Principles of Database Systems}, PODS '17, page 37–52, New York, NY, USA,
  2017. Association for Computing Machinery.
\newblock \href {https://doi.org/10.1145/3034786.3034796}
  {\path{doi:10.1145/3034786.3034796}}.

\bibitem{Coward2023AutomatingCD}
Samuel Coward, George~A. Constantinides, and Theo Drane.
\newblock Automating constraint-aware datapath optimization using e-graphs.
\newblock {\em 2023 60th ACM/IEEE Design Automation Conference (DAC)}, pages
  1--6, 2023.
\newblock URL: \url{https://api.semanticscholar.org/CorpusID:257353847}.

\bibitem{chase-revisited}
Alin Deutsch, Alan Nash, and Jeff Remmel.
\newblock The chase revisited.
\newblock In {\em Proceedings of the Twenty-Seventh ACM SIGMOD-SIGACT-SIGART
  Symposium on Principles of Database Systems}, PODS '08, page 149–158, New
  York, NY, USA, 2008. Association for Computing Machinery.
\newblock \href {https://doi.org/10.1145/1376916.1376938}
  {\path{doi:10.1145/1376916.1376938}}.

\bibitem{risinglight}
RisingLight developers.
\newblock {RisingLight: An Educational OLAP Database System}, 12 2022.
\newblock URL: \url{https://github.com/risinglightdb/risinglight}.

\bibitem{tarjan-congruence}
{Peter J.} Downey, Ravi Sethi, and {Robert Endre} Tarjan.
\newblock Variations on the common subexpression problem.
\newblock {\em Journal of the ACM}, 27(4):758--771, October 1980.
\newblock \href {https://doi.org/10.1145/322217.322228}
  {\path{doi:10.1145/322217.322228}}.

\bibitem{FAGIN200589}
Ronald Fagin, Phokion~G. Kolaitis, Renée~J. Miller, and Lucian Popa.
\newblock Data exchange: semantics and query answering.
\newblock {\em Theoretical Computer Science}, 336(1):89--124, 2005.
\newblock Database Theory.
\newblock URL:
  \url{https://www.sciencedirect.com/science/article/pii/S030439750400725X},
  \href {https://doi.org/https://doi.org/10.1016/j.tcs.2004.10.033}
  {\path{doi:https://doi.org/10.1016/j.tcs.2004.10.033}}.

\bibitem{tac-termination}
Thomas Genet.
\newblock Termination criteria for tree automata completion.
\newblock {\em Journal of Logical and Algebraic Methods in Programming}, 85(1,
  Part 1):3--33, 2016.
\newblock Rewriting Logic and its Applications.
\newblock URL:
  \url{https://www.sciencedirect.com/science/article/pii/S2352220815000504},
  \href {https://doi.org/https://doi.org/10.1016/j.jlamp.2015.05.003}
  {\path{doi:https://doi.org/10.1016/j.jlamp.2015.05.003}}.

\bibitem{gilleron1995regular}
R\'{e}my Gilleron and Sophie Tison.
\newblock Regular tree languages and rewrite systems.
\newblock {\em Fundam. Inf.}, 24(1–2):157–175, apr 1995.

\bibitem{gogacz2014allinstance}
Tomasz Gogacz and Jerzy Marcinkowski.
\newblock All--instances termination of chase is undecidable.
\newblock In Javier Esparza, Pierre Fraigniaud, Thore Husfeldt, and Elias
  Koutsoupias, editors, {\em Automata, Languages, and Programming}, pages
  293--304, Berlin, Heidelberg, 2014. Springer Berlin Heidelberg.

\bibitem{grahne18anatomy}
G{\"{o}}sta Grahne and Adrian Onet.
\newblock Anatomy of the chase.
\newblock {\em Fundam. Informaticae}, 157(3):221--270, 2018.
\newblock \href {https://doi.org/10.3233/FI-2018-1627}
  {\path{doi:10.3233/FI-2018-1627}}.

\bibitem{gulwani2005join}
Sumit Gulwani, Ashish Tiwari, and George~C. Necula.
\newblock Join algorithms for the theory of uninterpreted functions.
\newblock In Kamal Lodaya and Meena Mahajan, editors, {\em FSTTCS 2004:
  Foundations of Software Technology and Theoretical Computer Science}, pages
  311--323, Berlin, Heidelberg, 2005. Springer Berlin Heidelberg.

\bibitem{DBLP:journals/ejc/HellN91}
Pavol Hell and Jaroslav Nesetril.
\newblock Images of rigid digraphs.
\newblock {\em Eur. J. Comb.}, 12(1):33--42, 1991.
\newblock \href {https://doi.org/10.1016/S0195-6698(13)80005-4}
  {\path{doi:10.1016/S0195-6698(13)80005-4}}.

\bibitem{denali}
Rajeev Joshi, Greg Nelson, and Keith Randall.
\newblock Denali: A goal-directed superoptimizer.
\newblock {\em SIGPLAN Not.}, 37(5):304--314, May 2002.
\newblock URL: \url{http://doi.acm.org/10.1145/543552.512566}, \href
  {https://doi.org/10.1145/543552.512566} {\path{doi:10.1145/543552.512566}}.

\bibitem{DBLP:journals/ipl/Knuth77}
Donald~E. Knuth.
\newblock A generalization of dijkstra's algorithm.
\newblock {\em Inf. Process. Lett.}, 6(1):1--5, 1977.
\newblock \href {https://doi.org/10.1016/0020-0190(77)90002-3}
  {\path{doi:10.1016/0020-0190(77)90002-3}}.

\bibitem{kozen1992myhill}
Dexter Kozen.
\newblock On the myhill-nerode theorem for trees.
\newblock {\em Bulletin of the EATCS}, 47:170--173, 1992.

\bibitem{kozen1993partial}
Dexter Kozen.
\newblock {\em Partial Automata and Finitely Generated Congruences: An
  Extension of Nerode's Theorem}, pages 490--511.
\newblock Birkh{\"a}user Boston, Boston, MA, 1993.
\newblock \href {https://doi.org/10.1007/978-1-4612-0325-4_16}
  {\path{doi:10.1007/978-1-4612-0325-4_16}}.

\bibitem{marnette2009generalized}
Bruno Marnette.
\newblock Generalized schema-mappings: From termination to tractability.
\newblock In {\em Proceedings of the Twenty-Eighth ACM SIGMOD-SIGACT-SIGART
  Symposium on Principles of Database Systems}, PODS '09, page 13–22, New
  York, NY, USA, 2009. Association for Computing Machinery.
\newblock \href {https://doi.org/10.1145/1559795.1559799}
  {\path{doi:10.1145/1559795.1559799}}.

\bibitem{efficient-e-matching}
Leonardo Moura and Nikolaj Bj\o{}rner.
\newblock Efficient e-matching for smt solvers.
\newblock In {\em Proceedings of the 21st International Conference on Automated
  Deduction: Automated Deduction}, CADE-21, page 183–198, Berlin, Heidelberg,
  2007. Springer-Verlag.
\newblock URL: \url{https://doi.org/10.1007/978-3-540-73595-3_13}, \href
  {https://doi.org/10.1007/978-3-540-73595-3\_13}
  {\path{doi:10.1007/978-3-540-73595-3\_13}}.

\bibitem{szalinski}
Chandrakana Nandi, Max Willsey, Adam Anderson, James~R. Wilcox, Eva Darulova,
  Dan Grossman, and Zachary Tatlock.
\newblock Synthesizing structured {CAD} models with equality saturation and
  inverse transformations.
\newblock In {\em Proceedings of the 41st ACM SIGPLAN Conference on Programming
  Language Design and Implementation}, PLDI 2020, page 31–44, New York, NY,
  USA, 2020. Association for Computing Machinery.
\newblock \href {https://doi.org/10.1145/3385412.3386012}
  {\path{doi:10.1145/3385412.3386012}}.

\bibitem{narendran1985complexity}
Paliath Narendran, Colm Ó'Dúnlaing, and Heinrich Rolletschek.
\newblock Complexity of certain decision problems about congruential languages.
\newblock {\em Journal of Computer and System Sciences}, 30(3):343--358, 1985.
\newblock URL:
  \url{https://www.sciencedirect.com/science/article/pii/0022000085900510},
  \href {https://doi.org/https://doi.org/10.1016/0022-0000(85)90051-0}
  {\path{doi:https://doi.org/10.1016/0022-0000(85)90051-0}}.

\bibitem{nelson}
Charles~Gregory Nelson.
\newblock {\em Techniques for Program Verification}.
\newblock PhD thesis, Stanford University, Stanford, CA, USA, 1980.
\newblock AAI8011683.

\bibitem{DBLP:journals/sigmod/NgoRR13}
Hung~Q. Ngo, Christopher R{\'{e}}, and Atri Rudra.
\newblock Skew strikes back: new developments in the theory of join algorithms.
\newblock {\em {SIGMOD} Rec.}, 42(4):5--16, 2013.
\newblock \href {https://doi.org/10.1145/2590989.2590991}
  {\path{doi:10.1145/2590989.2590991}}.

\bibitem{herbie}
Pavel Panchekha, Alex Sanchez-Stern, James~R. Wilcox, and Zachary Tatlock.
\newblock Automatically improving accuracy for floating point expressions.
\newblock {\em SIGPLAN Not.}, 50(6):1–11, June 2015.
\newblock \href {https://doi.org/10.1145/2813885.2737959}
  {\path{doi:10.1145/2813885.2737959}}.

\bibitem{storel}
Maximilian Schleich, Amir Shaikhha, and Dan Suciu.
\newblock Optimizing tensor programs on flexible storage, 2022.
\newblock URL: \url{https://arxiv.org/abs/2210.06267}, \href
  {https://doi.org/10.48550/ARXIV.2210.06267}
  {\path{doi:10.48550/ARXIV.2210.06267}}.

\bibitem{snyder1993rgrs}
Wayne Snyder.
\newblock A fast algorithm for generating reduced ground rewriting systems from
  a set of ground equations.
\newblock {\em J. Symb. Comput.}, 15(4):415–450, apr 1993.
\newblock \href {https://doi.org/10.1006/jsco.1993.1029}
  {\path{doi:10.1006/jsco.1993.1029}}.

\bibitem{eqsat}
Ross Tate, Michael Stepp, Zachary Tatlock, and Sorin Lerner.
\newblock Equality saturation: A new approach to optimization.
\newblock In {\em Proceedings of the 36th Annual ACM SIGPLAN-SIGACT Symposium
  on Principles of Programming Languages}, POPL '09, pages 264--276, New York,
  NY, USA, 2009. ACM.
\newblock URL: \url{http://doi.acm.org/10.1145/1480881.1480915}, \href
  {https://doi.org/10.1145/1480881.1480915}
  {\path{doi:10.1145/1480881.1480915}}.

\bibitem{diospyros}
Alexa VanHattum, Rachit Nigam, Vincent~T. Lee, James Bornholt, and Adrian
  Sampson.
\newblock {\em Vectorization for Digital Signal Processors via Equality
  Saturation}, page 874–886.
\newblock Association for Computing Machinery, New York, NY, USA, 2021.
\newblock URL: \url{https://doi.org/10.1145/3445814.3446707}.

\bibitem{spores}
Yisu~Remy Wang, Shana Hutchison, Jonathan Leang, Bill Howe, and Dan Suciu.
\newblock {SPORES}: Sum-product optimization via relational equality saturation
  for large scale linear algebra.
\newblock {\em Proceedings of the VLDB Endowment}, 2020.

\bibitem{wang2022optimizing}
Yisu~Remy Wang, Mahmoud~Abo Khamis, Hung~Q Ngo, Reinhard Pichler, and Dan
  Suciu.
\newblock Optimizing recursive queries with program synthesis.
\newblock {\em arXiv preprint arXiv:2202.10390}, 2022.

\bibitem{egg}
Max Willsey, Chandrakana Nandi, Yisu~Remy Wang, Oliver Flatt, Zachary Tatlock,
  and Pavel Panchekha.
\newblock Egg: Fast and extensible equality saturation.
\newblock {\em Proc. ACM Program. Lang.}, 5(POPL), jan 2021.
\newblock \href {https://doi.org/10.1145/3434304} {\path{doi:10.1145/3434304}}.

\bibitem{tensat}
Yichen Yang, Phitchaya~Mangpo Phothilimtha, Yisu~Remy Wang, Max Willsey, Sudip
  Roy, and Jacques Pienaar.
\newblock Equality saturation for tensor graph superoptimization.
\newblock In {\em Proceedings of Machine Learning and Systems}, 2021.
\newblock \href {http://arxiv.org/abs/2101.01332} {\path{arXiv:2101.01332}}.

\bibitem{egglog}
Yihong Zhang, Yisu~Remy Wang, Oliver Flatt, David Cao, Philip Zucker, Eli
  Rosenthal, Zachary Tatlock, and Max Willsey.
\newblock Better together: Unifying datalog and equality saturation.
\newblock {\em Proc. ACM Program. Lang.}, 7(PLDI), jun 2023.
\newblock \href {https://doi.org/10.1145/3591239} {\path{doi:10.1145/3591239}}.

\bibitem{relational-ematching}
Yihong Zhang, Yisu~Remy Wang, Max Willsey, and Zachary Tatlock.
\newblock Relational e-matching.
\newblock {\em Proc. ACM Program. Lang.}, 6(POPL), jan 2022.
\newblock \href {https://doi.org/10.1145/3498696} {\path{doi:10.1145/3498696}}.

\end{thebibliography}

\appendix

\section{Proof basic properties of E-graphs and EqSat }

\subsection{Proof for \autoref{lem:rebuilding-exists}}
\label{apdx:rebuilding-exists}

\rebuildingexists*

\begin{proofsketch}
    Suppose $\calA = \langle Q, \Sigma, Q_\textit{final}, \Delta \rangle$ and let $L=\bigcup \{\mathcal{L}(c) \mid c\in \mathcal{A}\}$, the set of terms represented by any state of $\calA$. 
    Define $(\approx_L)\subseteq L\times L$
    as the smallest equivalence relation satisfying
    $t_1\approx_L t_2$ if some state $c$ accepts both $t_1$ and $t_2$.
    Let $C$ be the set of equivalence classes of $\approx_L$.
    Let 
    \[N=\setof{f([t_1]_{\approx_L}, \ldots, [t_k]_{\approx_L}\rightarrow [f(t_1, \ldots, t_k)]_{\approx_L})}{f(t_1,\ldots, t_k)\in L}.\]
  
    We claim that $G=\langle C, \Sigma, N\rangle$ is such an E-graph.
    $G$ is an E-graph as we can show $G$ is deterministic and reachable.
    We can also show that
    if two terms $t_1$ and $t_2$ are accepted by the same state of $\mathcal{A}$,
    it holds that $t_1\approx_L t_2$,
    and that if two terms $t_1\approx_L t_2$,
    $t_1$ and $t_2$ are accepted by the same E-class of $G$.
    Therefore, 
    terms accepted by the same state of $\mathcal{A}$ are accepted by the same E-class of $G$.
    Denote this mapping $h$.
    It can be shown that $h$ is a homomorphism from $\mathcal{A}$ and $G$.
  
    Let $G'$ be another E-graph and $u$ be a homomorphism from $\mathcal{A}$ to $G'$.
    We need to show there is a homomorphism from $G$ to $G'$.
    This can be done by noticing that 
    if $t_1\approx_L t_2$, $t_1$ and $t_2$ are accepted by the same E-class of $G'$.
    We can then define $h'([t]_{\approx_L})=c$ where $t$ is accepted by $c$ in $G'$
    and show $h'$ is a homomorphism from $G$ to $G'$.
    
    The uniqueness of $G$ follows from the anti-symmetry of $\sqsubseteq$.
  \end{proofsketch}

\subsection{Proof for \autoref{lemma:union}}
\label{apdx:least-upper-bound}

\leastupperbound*

\begin{proof}
    $G_i\sqsubseteq G$ since there is an identity homomorphism from $G_i$ to $\mathcal{A}_{\sqcup}$ 
    and a homomorphism, denoted as $h_\congr$, from $\mathcal{A}_{\sqcup}$ to $G$ by \autoref{lem:rebuilding-exists}.
  
    Next, we show if $G'$ is an E-graph such that $\forall i. G_i\sqsubseteq G'$,
    $G\sqsubseteq G'$.
    Denote the homomorphism from $G_i$ to $G'$ as $h_i$.
    There is a homomorphism from $\mathcal{A}_{\sqcup}$ to $G'$, 
    defined by $h(c) = h_i(c)$ for $c$ is an E-class of $G_i$.
    By \autoref{lem:rebuilding-exists},
    there is a homomorphism $h'$ from $G$ to $G'$.
      \[\begin{tikzcd}
        {G_i} & {\mathcal{A}_\sqcup=\langle\bigcup_i Q_i, \bigcup_i \Delta_i\rangle} & G \\
        && {G'}
        \arrow["{\textit{id}}", from=1-1, to=1-2]
        \arrow["{h_\congr}", from=1-2, to=1-3]
        \arrow["{h'}", dashed, from=1-3, to=2-3]
        \arrow["{h_i}"', from=1-1, to=2-3]
        \arrow["{h}"{description}, dashed, from=1-2, to=2-3]
      \end{tikzcd}\]
  \end{proof}

\subsection{Proof for \autoref{thm:eqsat:dfn}}
\label{apdx:eqsatdfn}
\eqsatdfnthm*

\begin{proofsketch}
  By definition, $H$ is a model of $\trs,G$ iff
  $G \sqsubseteq H$ and $H$ is a fixpoint of $\ico_{\trs}$.
  Therefore, it suffices to prove that
  $H \defeq \bigsqcup_{i\geq 0} \ico_{\trs}^{(i)}(G)$ is the least
  fixpoint.  We first prove that $H$ is a fixpoint.
  $H \sqsubseteq \ico_{\trs}(H)$ since $\ico_{\trs}$ is inflationary,
  so we prove the opposite.  Apply definition of $T_{\trs}$,
  $T_{\trs}(H) = H \cup (\bigcup \flatt)$, where $\bigcup \flatt$
  abbreviates the union of E-matches into $H$. Since $\trs$ is finite
  and every E-match $\sigma:\varset(\lhs) \rightarrow Q$ uses a finite
  number of states, there exists an $i$ such that all E-matches use
  only the states of $\ico_{\trs}^{(i)}(G)$.  This implies
  $(\bigcup \flatt)\sqsubseteq \ico_{\trs}^{(i)}(G) \sqsubseteq H$,
  proving that $T_{\trs}(H) \sqsubseteq H$.  It follows that
  $\ico_{\trs}(H) = \congr(T_{\trs}(H)) \sqsubseteq \congr(H) = H$
  because $H$ is deterministic.  Thus, $H$ is a fixpoint.  We prove
  that it is the least: let $H'$ be another fixpoint.  We use
  induction on $i$ to prove $\ico_{\trs}^{(i)}(G) \sqsubseteq H'$:
  assuming this holds for $i$, we derive
  $\ico_{\trs}(\ico_{\trs}^{(i)}(G)) \sqsubseteq \ico_{\trs}(H') =H'$
  thus it holds for $i+1$.  Therefore
  $H = \bigsqcup_{i \geq 0}\ico_{\trs}^{(i)}(G) \sqsubseteq H'$,
  completing the proof.
\end{proofsketch}

\section{Proof for the finite convergence lemma}
\label{apdx:convergence}

\convergence*

\begin{proof}
  For the sake of contradiction we assume $\mathcal{G}$ is infinite.
  Let us first the define the \textit{rank} of an E-class as follows.
    \begin{align*}
  \textit{rank}_G(c)& = \min_{\left(f\left(c_1,\ldots, c_n\right)\rightarrow c\right)\in N} \textit{rank}_G(f(c_1, \ldots, c_n))\\
  \textit{rank}_G(f(c_1,\ldots, c_n)) &= 1 + \max\{0, \textit{rank}_G(c_1), \ldots \textit{rank}_G(c_n)\}
  \end{align*}
  Intuitively, the rank of an E-class is the smallest depth of the terms represented by this E-class.
  By E-graph's reachability, every E-class has a finite rank.
  We define the rank of an E-graph as the greatest rank of its E-classes.
  Some observations about E-graph ranks:

  \begin{fac}
    Every finite E-graph has a finite rank.
  \end{fac}
  \begin{fac}\label{fact:finite-bounded-rank}
    The set of E-graphs with a bounded rank is finite.\footnote{To see this, E-graph rank bounds the number of E-classes an E-graph can have.}
  \end{fac}
  \begin{fac}\label{fact:rank-decreases}
    Ranks only decrease:
    Let $G$ and $H$ be two E-graphs and $h:G\rightarrow H$ a homomorphism. 
    $\textit{rank}_G(c)\geq \textit{rank}_H(h(c))$.
  \end{fac}

  Let us denote the homomorphism between $G_i$ and $G_j$ as $h_{i,j}$ for all $i\leq j$,
  and the homomorphism between $G_i$ and $G_{\infty}$ as $h_{i,\infty}$.
  Since $G_{\infty}$ is finite, denote its rank as $N$.

  \newcommand{\unionn}{$\textit{Union}_N$}

  We sat a homomorphism $h:G\rightarrow H$ is \unionn{} if there exist 
  two E-classes $c_1$ and $c_2$ such that
  $\textit{rank}_G(c_1)\leq N$, $\textit{rank}_G(c_2)\leq N$,
  and $h(c_1)=h(c_2)$.
  In other words, $H$ ``unions'' two E-classes with rank $\leq N$ of $G$.
  We claim there exists some lower bound $M$ such that for all $M\leq i\leq j$, $h_{i,j}$ is not \unionn{}.
  This is because every \unionn{} homomorphism necessarily merges two equivalence classes with rank $\leq N$,
  and this cannot happen indefinitely because each such equivalence class need to be backed by a term with depth $\leq N$, which is finite.

  Find $a_1>M$ such that $G_{a_1}$ contains an E-class $c_1$ with rank $N+1$.
  This is always possible by \autoref{fact:finite-bounded-rank} and 
  the fact that any E-graph with rank $>N$ must contain some E-classes with rank $N+1$.
  Let $n_1=f_1({\vec{\mathbf{c}}_1})\rightarrow c_1$ be an E-node of E-class $c_1$
  such that $\max_i \textit{rank}_G(\left(\vec{\mathbf{c}_1}\right)_i)\leq N$.

  Since all E-classes of $G_{\infty}$ have rank $\leq N$, 
  there must exist some $b_1>a_1$ such that 
  $h_{a_1,b_1}(c)$ in $G_{b_1}$ has rank $\leq N$ .
  Next, find an $a_2>b_1$ such that $G_{a_2}$ contains an E-class $c_2$ with rank $N+1$.
  Repeat the same process to obtain a sequence $\mathcal{N}$ of 
  E-nodes
  \begin{align*}
    f_1({\vec{\mathbf{c}}_1})\rightarrow c_1\quad \in \quad G_{a_1}\\
    f_2({\vec{\mathbf{c}}_2})\rightarrow c_2\quad \in \quad G_{a_2}\\
    \ldots\quad\quad
  \end{align*}
  that satisfies $\textit{rank}_{G_{a_i}}(c_i)=N+1$, $\max_j \left(\vec{\mathbf{c}}_i\right)_j\leq N$,
  and $\textit{rank}_{G_{b_i}}(h{a_i, b_i}(c_i))\leq N$.
  Note by \autoref{fact:rank-decreases},
  it also holds that for all $k>b_i$, $h_{a_i, k}(c_i)$ also has rank $\leq N$.

  For the rest of the proof, we claim that for all $i$ and $j$, $h_{a_i,\infty}(f_i(\vec{\mathbf{c}}_i))\neq h_{a_j,\infty}(f_j(\vec{\mathbf{c}}_j))$,
  which implies $G_{\infty}$ has infinitely many distinct E-classes as $\mathcal{G}$ is infinite.
  This is a contradiction to the fact that $G_{\infty}$ is finite.

  To prove this claim, we assume $i<j$ without loss of generality.
  Observe that $h_{a_i, a_j}(c_i)$ has rank $\leq N$ and $c_j$ has rank $N+1$ in $G_{a_j}$,
  so \[
    h_{a_i, a_j}(c_i)\neq c_j.
    \]
  By determinacy of E-graphs,
  \[
    h_{a_i, a_j}(f_i(\vec{\mathbf{c}}_i))\neq f_j(\vec{\mathbf{c}}_j).
    \]
  Because E-classes in $\vec{\mathbf{c}}_i$ and $\vec{\mathbf{c}}_j$ has rank $\leq N$,
  it holds that for all $k\geq a_j$,
  \begin{align*}
    h_{a_j, k}(h_{a_i, a_j}(f_i(\vec{\mathbf{c}}_i)))\neq h_{a_j, k}(f_j(\vec{\mathbf{c}}_j))\\
    h_{a_i, k}(f_i(\vec{\mathbf{c}}_i))\neq h_{a_j, k}(f_j(\vec{\mathbf{c}}_j))
  \end{align*}
  since $h_{a_j, k}$ is not \unionn{}.
  This inequality is preserved \textit{ad infinitum}.
  By virtue of $G_{\infty}$ being the least upper bound,
  \[
    h_{a_i, \infty}(f_i(\vec{\mathbf{c}}_i))\neq h_{a_j, \infty}(f_j(\vec{\mathbf{c}}_j)).
  \]
  \end{proof}

\section{Proof for the reductions between EqSat and the chase}

\subsection{Proof for \autoref{thm:skolem-to-eqsat}}
\label{apdx:skolem-to-eqsat}

\skolemtoeqsatthm*

\begin{figure}
    \centering
    $\begin{array}{ccr}
        r_{\boldtop}:&\ 
          \boldtop \rightarrow \boldland(\boldtop, \boldtop)\\
        r_R:&\ 
          R(x_1,\ldots, x_n)\rightarrow \boldtop &\text{for relation symbol $R\in\sigma$}\\
        r_d:&\ 
          \begin{array}{@{}l@{}}
          \boldland\left(R_1\left(\vec{x_1},\vec{y_1}\right),\ldots \right)\rightarrow \\
          \boldland\left(
            R'_1\left(\vec{x'_1}, \vec{f^d_z}_1\right),
          \ldots\right)
          \end{array}
        &\text{for dep.~
          $d:\begin{array}{c}
            R_1\left(\vec{x_1},\vec{y_1}\right)\land \ldots\rightarrow \\ 
            R'_1\left(\vec{x'_1},\vec{f^d_z}_1\right)\land\ldots\end{array}\in \Gamma$}\\
        r_t:&\ \boldtop \rightarrow R(c_1, \ldots, c_n)
         & \text{for initial tuple $t=R(c_1,\ldots, c_n)\in I$}
    \end{array}$
    \caption{A summary of encoded rewrite rules in \autoref{thm:skolem-to-eqsat} 
    for the given $\Gamma$ and $I$. The initial term to equality saturation is $\boldtop$.}
    \label{fig:skolem-rule-summary}
\end{figure}

\begin{proof}
    Denote $G=\eqsat(\trs, t)$ and $I=\sklchase(\Gamma, I)$.
    Denote the E-class of $G$ that represents $\boldtop$ as $c_{\boldtop}$.

    We show the following holds between $G$ and $I$:
    \begin{center}
        A tuple $R(c_1, \ldots, c_n)$ is in $I$ if and only if
        $R(c_1, \ldots, c_n)$ is represented by $G$,
        where $R$ is a function symbol in $\Sigma$
        and $c_1,\ldots, c_n$ are Skolem terms or constants.
    \end{center}

    \begin{itemize}
        \item $\Rightarrow$: We prove by induction on
        the iteration of applying match/apply operator $T_\trs$ when a tuple is first derived in $I$.
        \begin{itemize}
            \item Base case: this is clear from \autoref{fig:skolem-to-eqsat}.
            Denote the E-graph after applying the first iteration of equality saturation as
            $G_1$.
            In the first iteration of equality saturation,
            rules that are applied are $\boldtop\rightarrow \boldland(\boldtop, \boldtop)$ and
            $\boldtop\rightarrow R(c_1, \ldots, c_n)$ for each tuple in the initial database,
            so every tuple $R(c_1, \ldots, c_n)$ in the input database is represented $G_1$.
            Since $G_1\sqsubseteq G$, every tuple is also represented by $G$.
            \item Inductive case: 
            Assume a tuple $R_i(h(\vec{x'_i}), h(\vec{f^d_z}_i))$ is produced at the $n$-th iteration of $T_\trs$
            by dependency 
            \[
                d: \ R_1\left(\vec{x_1},\vec{y_1}\right)\land\ldots \land R_n\left(\vec{x_n}, \vec{y_n}\right)
            \rightarrow 
            R'_1\left(\vec{x'_1}, \vec{f^d_z}_1\right)\land \ldots\land R'_m\left(\vec{x'_m}, \vec{f^d_z}_m\right)
            \]
            and substitution $h$.
            By inductive hypothesis, every tuple $R_i(h(\vec{x_i}), h(\vec{y_i}))$
            in the substituted body of $d$ is also represented by $c_{\boldtop}$ of $G$.

            It is straightforward to see that $G$ contains E-nodes $\boldtop\rightarrow c_{\boldtop}$ and 
            $\boldland(\boldtop, \boldtop)\rightarrow c_{\boldtop}$.
            As a result,
            \begin{align*}
                \boldland\left(R_1\left(h\left(\vec{x_1}\right), h\left(\vec{y_1}\right)\right), \ldots, 
                \boldland\left(R_n\left(h\left(\vec{x_n}\right), h\left(\vec{y_n}\right)\right),\boldtop \right)\right)
                \rightarrow^* &
                \boldland\left(c_{\boldtop}, \ldots, \boldland\left(c_{\boldtop},c_{\boldtop} \right)\right)
                \\
                \rightarrow^* &
                \;c_{\boldtop}
            \end{align*}
            Therefore, to reach fixpoint,
            equality saturation would also add the substituted right-hand side
            \[
                \boldland\left(R'_1\left(\vec{x'_1}, \vec{f^d_z}_1\right), \ldots,\boldland\left( R'_m\left(\vec{x'_m}, \vec{f^d_z}_m\right), \boldtop\right)\right)
            \]
            to $c_{\boldtop}$.
            Similarly, equality saturation would fire series of rules of the form $r_R:R(x_1, \ldots, x_n)\rightarrow \boldtop$,
            so $R_i(h(\vec{x'_i}), h(\vec{f^d_z}_i))$ would be represented by $c_{\boldtop}$ of $G$.
        \end{itemize}
    \item $\Leftarrow$: 
    It can be shown that every term $R(\vec{c})$ that is inserted into the E-graph at iteration $n$
    will be unioned with the $c_{\boldtop}$ at iteration $n+1$.
    Therefore, we only need to show that for every term $R(\vec{c})$ represented by some E-class of $G$,
    $R(\vec{c})$ is in $I$.
    Similar to the first case, 
    let us prove by inducting on the first time a $R$-term $R(\vec{c})$ is inserted into the E-graph.
    \begin{itemize}
        \item Base case: At iteration 1, all terms $R(c_1, \ldots, c_n)$ inserted correspond to tuples from the initial database.
        \item Inductive case: Suppose at iteration $i>1$, a $R$-term $R(\vec{c})[\sigma]$ is inserted into the E-graph
        because of some rewrite rule $r$ and E-class substitution $\sigma$.
        By definition, the rewrite rule $r$ has to be of the form 
        \[
            r_d:\boldland\left(R_1\left(\vec{x_1},\vec{y_1}\right),
            \ldots
            \right)
            \rightarrow 
            \boldland\left(
            R'_1\left(\vec{x'_1}, \vec{f^d_z}_1\right),
            \ldots \right).
     \]
     By inductive hypothesis, 
     all tuples $R_i(\vec{c_i})$ represented by $R_i(\vec{x_i}, \vec{y_i})[\sigma]$ are in $I$,
     so the dependency that $r$ is mapped from would be fired and 
     tuples represented by $R(\vec{c})[\sigma]$ are added to $I$.
    \end{itemize}
    \end{itemize}

    Next, we show that the Skolem chase terminates if and only if equality saturation terminates:
    if equality saturation terminates,
    an E-graph with a finite number of transitions is produced,
    so the mapped database is finite.
    Since the Skolem chase monotonically adds tuples to the database,
    the Skolem chase has to terminate as well.
    On the other hand, if the Skolem chase is finite,
    the produced E-graph needs to be finite as well, 
    since it only contains $\boldtop\rightarrow c_{\boldtop}$ 
    and $\boldland(\boldtop, \boldtop)\rightarrow c_{\boldtop}$
    besides E-nodes that are in one-to-one correspondence with the databases.
    By \autoref{lemma:convergence}, equality saturation terminates as well.
\end{proof}

\subsection{Proof for \autoref{thm:eqsat-to-chase}}
\label{apdx:eqsat-to-chase}

\eqsattochasethm*

To prove this theorem, we show the following lemma. Recall that a database is a core
 if every homomorphism from the database to itself is an isomorphism.
Two homomorphically equivalent cores are isomorphic.

\begin{lem}
  Given a set of rewrite rules $\Sigma$,
  let $I$ be a database instance that encodes an E-graph
  and $\Gamma$ be a set of dependencies encoding EqSat of $\Sigma$, both using the encoding in \autoref{sec:eqsat-to-chase}.
  Suppose $I_k$ is a database obtained after a finite number of chase steps of $\Gamma$ and $I$.
  If $I_k$ is closed under the EGD subset of $\Gamma$, $I_k$ is a core database instance.
\end{lem}

Consider the tree automaton $\mathcal{A}$ encoded by $I_k$. 
Because $I_k$ is closed under EGDs, $\mathcal{A}$ is deterministic.
Since every dependency of $\Sigma$ preserves reachability of the encoded E-graph,
$\mathcal{A}$ is reachable, so $\mathcal{A}$ is an E-graph.
Therefore, $\mathcal{A}$ is a core tree automaton.
The definition of homomorphisms is preserved under our encoding,
 so the only homomorphism from $I_k$ to $I_k$ is the identity mapping.
Therefore $I_k$ is a core.

\begin{cor}
  Fix $\Sigma$, $I$, $\Gamma$ as above. 
  Let $I_\infty$ be the result of a (potentially non-terminating) chase sequence of $\Sigma$ and $I$.
  $I_\infty$ is a core database instance.
\end{cor}

$I_\infty$ is closed under the EGD subset of $\Gamma$. By the above lemma we only need to consider
 the case where $I_\infty$ is the result of a non-terminating chase.
Since $I_\infty$ is closed under the EGDs, it encodes a deterministic tree automaton $\mathcal{A}$.
Suppose for the sake of contradiction that $\mathcal{A}$ is not reachable, 
 with some state $q$ that does not accept any term.
There must exist a finite number $m$ where $q$ stays unreachable for all $\geq m$ iterations.
However, this is impossible, since the head of each dependency of $\Gamma$ preserves reachability.

\begin{proof}[Proof of \autoref{thm:eqsat-to-chase}]
  We use the encoding in \autoref{sec:eqsat-to-chase} and
  treat every tuple $R_f(c_1, \ldots, c_n, c)$ as an E-node $f(c_1, \ldots, c_n)\rightarrow c$ 
  in an E-graph.
  It is straightforward to see that every insertion of the match/apply operator
  corresponds to a chase step of some TGD (EGD if $\rhs$ is a single-variable)
  of the above the encoding, 
  and every merge by $\congr$ corresponds to a chase step of some FD.
  Therefore, a terminating run of equality saturation 
  corresponds to a terminating chase sequence ((1)$\implies$ (2)).
  Denote this chase sequence as the EqSat-encoding chase sequence $\mathcal{I}_\eqsat$. 
  and the resulting database is isomorphic to the E-graph (when encoded as a database).

  Suppose there is a terminating chase sequence that produces some database instance $I$.
  By our lemma $I$ is core.
  Let $J$ be the database instance produced by $\mathcal{I}_\eqsat$.
  $J$ is also a core, since it encodes an E-graph.
  Because $I$ and $J$ are both universal models,  they are homomorphically equivalent.
  Homomorphically equivalent cores are isomorphic, $I$ and $J$ has to be isomorphic.
  Since $I$ is finite, $J$ is also finite.
  By \autoref{lemma:convergence},
  equality saturation also terminates ((2)$\implies$ (1)).

  The fact that (3)$\implies$ (1) is obvious, 
  since $\mathcal{I}_\eqsat$ is EGD-fair.
  To show (1)$\implies$ (3),
  let $\mathcal{J}:J_0, J_1\ldots$ be an arbitrary 
  EGD-fair chase sequence with result $J_\infty$,
  and consider the EqSat-encoding chase sequence $\mathcal{I}_\eqsat:I_0,\ldots, I_n$.
  We claim that for all $i$, there exists a finite $j$ such that
  there is a homomorphism from $I_i$ to $J_j$ and $J_j$ is closed under EGDs.
  This can be proved by induction on $i$.
  The case for $i=0$ holds trivially.
  Suppose $I_i\sqsubseteq J_j$ for some $j$. The fact that
  some $j'$ exists such that there is a homomorphism from $I_{i+1}$ to $J_{j'}$
  follows from fairness of chase sequences.
  With EGD-fairness, it is then possible to find another $j''>j'$
  such that $J_{j''}$ is closed under EGDs and $J_j\sqsubseteq J_{j'} \sqsubseteq J_{j''}$.

  As a result, some $m$ exists such that 
  there is a homomorphism from $I_n$ to $J_m$.
  There is also a homomorphism from $J_m$ to $J_\infty$ 
  from the chase sequence and a homomorphism from $J_\infty$ to $I_n$ 
  since $J_\infty$ is a universal model.
  By our lemma $I_n$, $J_m$, $J_\infty$ are all core models.
  Since they are homomorphically equivalent, they have to be isomorphic.
  so $J_m$ and $J_\infty$ coincide.
  Therefore, The chase sequence $\mathcal{J}:J_0, J_1,\ldots,J_m$ stops in $m$ steps
  ((1)$\implies$ (3)).

\[\begin{tikzcd}
	{I_0} & \ldots & {I_n} \\
	{J_0} & \ldots & {J_m} & \ldots & {J_\infty}
	\arrow[from=1-1, to=1-2]
	\arrow[from=1-2, to=1-3]
	\arrow[from=2-1, to=2-2]
	\arrow[from=2-2, to=2-3]
	\arrow[from=1-3, to=2-3]
	\arrow[from=2-3, to=2-4]
	\arrow[from=2-4, to=2-5]
	\arrow[from=2-5, to=1-3]
	\arrow[Rightarrow, no head, from=1-1, to=2-1]
\end{tikzcd}\]

It is left to show that if EqSat terminates, any arbitrary chase sequence,
terminating or not, has result
isomorphic to the output of equality saturation.
This can be shown similarly by noticing that $J_\infty$ 
 is a core model and homomorphically equivalent to the result of EqSat,
 so they have to be isomorphic.

\end{proof}

\section{Proof for the termination theorems}

Before we proceed to the proofs, let us complement the backgrounds on term rewriting system.
Given a term rewriting system $\trs$, a normal form is a term that cannot be rewritten any further.
We say $n$ is a normal form of $t$ if $t$ reduces to $n$ and $n$ is a normal form.
A TRS $\trs$ is \emph{terminating} if there is no infinite rewriting chain $t_1\rightarrow_{\trs} t_2\rightarrow \ldots$.
A TRS $\trs$ is \emph{confluent} if for all $t, t_1, t_2$, 
$t_1\leftarrow_{\trs}^*t\rightarrow_{\trs}^*t_2$ implies there exists a $t'$ such that $t_1\rightarrow_{\trs}^* t'\leftarrow_{\trs}^* t_2$. 
We call a confluent and terminating TRS \emph{convergent}.
Every term in a terminating TRS has at least one normal form,
 every term in a confluent TRS has at most one normal form,
 and every term in a convergent TRS has exactly one normal form.

\subsection{Proof for \autoref{thm:termination} and  \autoref{thm:all-term-instance}}
\label{apdx:termination}

\terminationthm*

This problem is in R.E.~since we can simply run
EqSat with \(\trs\) and \(t\) to test whether it terminates. To show this
problem is R.E.-hard, we reduce the halting problem of Turing
machines to the termination of EqSat. The proof follows that of
Narendran et al.~\cite{narendran1985complexity}.

In this proof, we consider a degenerate form of EqSat that works with
\emph{string} rewriting systems instead of TRS. A string can be viewed as a degenerate term, and a string rewriting rule
can be viewed as
a degenerate term rewriting rule. For example, the string \(uvw\) corresponds to a term
\(u(v(w(\epsilon)))\), where \(u, v, w\) are unary
functions and \(\epsilon\) is a special constant marking the end of a string. 
A string rewriting rule
\(uvw\rightarrow vuw\) corresponds to a (linear) term rewriting rule
\(u(v(w(x)))\rightarrow v(u(w(x)))\) where \(x\) is a variable.

For each Turing machine
\(\mathcal{M}\), we produce a string rewriting system \(\trs\) such that
the congruence closure of \(R\),
\((\approx_\trs)\), satisfies that each
congruent class of \(\approx_\trs\) corresponds to a trace of the Turing
machine. As a result, informally, the following statements are equivalent:

\begin{enumerate}
\item
  the Turing machine halts;
\item
  the trace of the Turing machine is finite;
\item
  the congruent class in \(\approx_\trs\) is finite;
\item
  EqSat terminates.
\end{enumerate}

A first simple idea is to encode the transition relation of a Turing machine
 directly as a string rewriting system,
 so that the congruence class of the initial configuration $w_0$ contains the trace.
For example, consider transition \(q_iabRq_j\), 
 which says if the current state if $q_i$ and the symbol being scanned is $a$, then write $b$ to the tape, move right, and change the state to $q_j$.
It is tempting to encode this transition as a string rewriting rule
 $q_ia\rightarrow_\trs b q_j$.
The issue, however, is that two different initial configurations
 can lead to the same configuration.
For example, consider a Turing machine that clears its input and then halts.
Every input string leads to the same configuration, 
 so its termination on an input does not imply the finiteness of its congruent class.

To address this issue, following Narendran et al.~\cite{narendran1985complexity}, 
 we require the string rewrite rule not only encode the transition relation,
 but also stores the history of the computation.
As a result, even if two different initial configurations lead to the same configuration,
 the rewritten strings that correspond to the same configuration are different,
 so different initial configurations lead to different congruent classes.
More specifically,
 we introduce dummy symbols that stores states and symbols before transitions.

\mysubparagraph{Turing machine}

A Turing machine $\mathcal{M}=(Q,\Sigma, \Pi,\Delta,q_0,\beta)$ consists of a set of states $Q$, 
 the input and the tape alphabet $\Sigma$ and $\Pi$ (with $\Sigma\subseteq \Pi$), 
 a set of transitions $\Delta$, an initial state $q_0\in Q$,
 and a special blank symbol $\beta\in\Pi$. Each transition in $\Delta$ is a quintuple in 
 $Q\times \Pi\times \Pi\times \{L,R\} \times Q$.
For example, transition $q_iabRq_j$ means if the current state is $q_i$ and the symbol
 being scanned is $a$, then replace $a$ with $b$, move the head to the right, 
 and transit to state $q_j$.
We assume the Turing machine is two-way infinite
 (so that the head can move in both directions indefinitely)
 and deterministic.
Each configuration of $\mathcal{M}$ can be represented as $\rhd uq_i v \lhd$,
 where $\rhd$,$\lhd$ are left and right end  markers, 
 $u$ is the string to the left of the read/write head
 $q_i$ is the current state,
 $v$ is the string to the right.
If $v$ is non-empty, the first character of $v$ is the symbol being scanned.
Otherwise, the symbol being scanned is the blank symbol $\beta$.
We say $w_1\vdash_{\mathcal{M}} w_2$ if configuration $w_1$ can transition 
 to configuration $w_2$ in a Turing machine $\mathcal{M}$, and we omit $\mathcal{M}$ when it's clear from the context.

We say a configuration $w$ is {\it halting} configuration
  if it cannot be transited further.
We say a configuration $w$  is {\it mortal} if there exists a finite sequence of configurations
  $w=w_0\vdash w_1\vdash\ldots\vdash w_n$ such that $w_n$ is a halting configuration.

\mysubparagraph{Alphabet of the constructed string rewriting system}

Compared to a direct encoding of the transition relation,
 the string rewriting system $\trs$ we construct has the following characteristics:

\begin{itemize}
  \item It distinguishes between symbols to the left and to the right of the current state:
  Define $\overline \Pi=\setof{\overline a}{a\in \Pi}$ as the alphabet used exclusively to the left of the state (and $\Pi$ to the right).
  \item We introduce new dummy symbol to store information about the states and symbols.
  Define $D_L=\setof{L_z}{z\in\overline \Pi\cup \{\lhd\}\text{ or } z\in \{\rhd\}\cup \overline\Pi}$ and $D_R$ similarly.
  Similar to $\overline \Pi$ and $\Pi$, $D_L$ (resp.~$D_R$) is 
  used exclusively for dummy symbols to the left (resp.~right) of the state symbol.
  \item Similar to $q_i\in Q$, of which the symbol in $\Pi$ to the immediate right
  is being scanned, we also define the ``left'' counterpart $\overline Q_i = \setof{\overline{q_i}}{q_i\in Q_i}$,
  where the symbol being scanned is to the immediate left of $\overline{q_i}$.
  For instance, in the string representation $\rhd u\overline{a q_i}v\lhd$, the current state is $q_i$ 
  and the symbol being scanned $a$.
\end{itemize}


In summary, the rewriting system we construct works over the following 
 regular string language
\[
  \configlang{}=
  \rhd (\overline\Pi\cup D_L)^*
  (Q\cup \overline Q)
  (\Pi\cup D_\trs)^*\lhd
\]
Strings in \configlang{} are in a many-to-one mapping to configurations
of a Turing machine. We denote this mapping as \(\pi\): \(\pi(w)\)
converts each \(\overline a\overline{q_i}\) to \(q_ia\) 
(and $\rhd\overline{q_i}$ to $\rhd q_i\beta$),
removes dummy
symbols \(L_z\) and \(R_z\), and replace \(\overline a\) with \(a\). For
example,
\[\pi(\rhd L_{q_0,a}\overline{b} L_{q_1,b} \overline{cq_3}dR_{q_i,\lhd}\lhd)=\rhd bq_3cd\lhd\]

Now, for each transitions in \(\mathcal{M}\), our string
rewriting system \(R\) is defined in \autoref{fig:tm-rewrite-rule}.
\begin{figure}
  \centering
\[
\begin{array}{|c|c|}
    \hline
\text{transitions in \(\mathcal{M}\)}
&
\text{rewrites in \(R\)}
\\
\hline
q_iabRq_j & q_ia\rightarrow_\trs L_{q_i,a}\overline bq_j \\
&
\overline a\overline{q_i}\rightarrow_\trs L_{\overline a,\overline{q_i}}\overline bq_j \\
\hline
q_i\beta bRq_j &
q_i\lhd\rightarrow_\trs L_{q_i,\lhd}\overline b q_j\lhd \\
&
\rhd \overline{q_i}\rightarrow_\trs \rhd L_{\rhd,\overline{q_i}}\overline bq_j \\
\hline
q_iabLq_j & q_ia\rightarrow_\trs \overline{q_j} b R_{q_i,a} \\
&
\overline a\overline{q_i}\rightarrow_\trs \overline{q_j} b R_{\overline a,\overline{q_i}} \\
\hline
q_i\beta bLq_j &
q_i\lhd\rightarrow_\trs \overline{q_j}bR_{q_i,\lhd}\lhd \\
&
\rhd \overline{q_i}\rightarrow_\trs \rhd\overline{q_j} b R_{\rhd,\overline{q_i}} \\
\hline
\end{array}
\quad
\cup 
\quad 
\begin{array}{|c|c|}
  \hline
  \text{for each $z$}
  &
  q_iR_z\rightarrow_\trs L_zL_zq_i\\
  & L_z\overline{q_i}\rightarrow_\trs \overline{q_i}R_zR_z \\
  \hline
\end{array}
\]
\caption{The string rewriting system \(R\) derived from a Turing machine \(\mathcal{M}\).}
\label{fig:tm-rewrite-rule}
\end{figure}
It consists of two parts.
The first part encodes each transition rule of \(\mathcal{M}\)
 as two rewrite rules, one for the case when the state symbol is $q_i\in Q$,
 and one for the case when the state symbol is $\overline{q_i}\in \overline{Q}$.
It also introduces dummy symbols to store the the source configuration of each transition.
Moreover, for each \(z\), we have the two additional sets of
rewrite rules 
\begin{align*}
q_iR_z&\rightarrow_\trs L_zL_zq_i\\
L_z\overline{q_i}&\rightarrow_\trs \overline{q_i}R_zR_z.
\end{align*}
They are used to ensure that the state symbol is adjacent to the symbol being scanned 
 and shuffling dummy symbols around.

To explain what these dummy symbol--shuffling rules do more precisely, let us define two types of strings of \configlang.
\begin{dfn}
  Type-A strings are strings in \configlang{}
  where the symbol being scanned is to the
  immediate right of \(q_i\) or to the immediate left of
  \(\overline {q_i}\). 
  In other words, we call a string \(s\) a type-A
  string if \(s\) contains \(q_ia\) or \(\overline{aq_i}\). 
  Type-B strings
  are strings in \configlang{} that are not type-A. 
\end{dfn}
The rewrite
rules above convert any type-B strings into type-A in a finite number of
steps.

Now, we observe that the string rewriting system \(R\) we constructed above has several properties:

\begin{enumerate}
\item
  Reverse convergence: the critical pair lemma states that if a
  rewriting system is terminating and all its critical pairs are
  convergent, it is convergent. Define \(R^{-1}\) to be a string rewriting system derived
  from \(R\) by swapping left- and right-hand side of each rewrite rule. 
  \(R^{-1}\) is
  terminating since rewrite rules in \(R^{-1}\) decreases the sizes of
  terms, and \(R^{-1}\) has no critical pairs. Therefore, \(R^{-1}\) is
  convergent.
\item
  For each type-A string \(w\), then either

  \begin{itemize}
  \item
    there exists no \(w'\) with \(w\rightarrow_\trs w'\) and \(\pi(w)\) is
    a halting configuration;
  \item
    there exists a unique \(w'\) such that \(w\rightarrow_\trs w'\).
    Moreover, it holds that
    $w'\in\configlang$ and \(\pi(w)\vdash \pi(w')\).
  \end{itemize}
\item
  For each type-B string \(w\), there exists a unique \(w'\) such that
  \(w\rightarrow_\trs w'\). It holds that $w'$ is in \configlang{} and \(\pi(w)=\pi(w')\).
 Moreover, if
  \(w_0\rightarrow_\trs w_1\rightarrow_\trs\ldots\) is a sequence of type-B
  strings, the sequence must be bounded in length, since the state
  symbols \(q_i\) and \(\overline{q_i}\) move towards one end according
  to the auxillary rules above.
\item
  From 2 and 3, it follows that $\rightarrow_\trs$ closed under 
  $\configlang{}$ (i.e., $w\in\configlang$ and $w\rightarrow_\trs w'$ implies $w'\in\configlang$)
  and is deterministic over \configlang{} (i.e.,
  \(w\rightarrow_\trs w_1\) and \(w\rightarrow_\trs w_2\) implies
  \(w_1=w_2\)).
\end{enumerate}

These observations allow us to prove the following lemma

\begin{lem}\label{lem:tm-cong-finite}
Given a Turing machine
\(\mathcal{M}\), construct a string rewriting system \(R\) as above.
Let $w_0$ be a string in \configlang{}.
\(\pi(w_0)\) is a mortal configuration of \(\mathcal{M}\) if and only if \([w_0]_\trs\), the
equivalence class of \(w_0\) in \(R\), is finite.
\end{lem}

\begin{proof}
  Without loss of generality,
  we assume $w_0$ is a normal form with respect to $R^{-1}$.
  If this is not the case, since $R^{-1}$ is convergent, $w_0$ has a unique normal form $w_0'$.
  Moreover, $[w_0]=[w'_0]$ and $\pi(w'_0)\vdash_{\mathcal{M}}^*\pi(w_0')$, 
  so it suffices to consider $w_0'$.

  \begin{itemize}
  \item
    \(\Leftarrow\): Suppose \([w_0]_\trs\) is finite.
    Since the size of each rewrite rule is strictly increasing,
    no cycle in the rewriting sequence is possible.
    Therefore, there must exist a finite sequence of
    \(W:w_0\rightarrow_\trs w_1\rightarrow_\trs \ldots \rightarrow_\trs w_n\) such
    that $w_n$ is a normal form of $\trs$.
  
    By our observation above, since $w_n$ cannot be rewritten further,
    \(w_n\) is a type-A string, and \(\pi(w_n)\) is a halting configuration.
  
    Take the subsequence of \(S\) consisting of all type-A strings:
    \[w_{a_0}\rightarrow_\trs^* w_{a_1}\rightarrow_\trs^*\ldots \rightarrow_\trs^*w_{a_k}=w_n.\]
    We have \(\pi(w_{a_i})\vdash\pi(w_{a_{i+1}})\) for all \(i\) and
    \(\pi(w_{a_k})\) is a halting configuration. 
    This implies a finite trace of the Turing machine:
    \[\pi(w_0)=\pi(w_{a_0})\vdash \pi(w_{a_1})\vdash\ldots\vdash \pi(w_{a_n}),\] 
    which implies  \(w_0\) is a mortal configuration of $\mathcal{M}$.
  \item
    \(\Rightarrow\): Suppose otherwise \(\pi(w_0)\) is a mortal configuration
    of \(\mathcal{M}\)
    and \([w_0]_\trs\) is infinite.
  
    Since \(w_0\) is a normal form with respect to
    \(\trs^{-1}\) and \(\trs^{-1}\) is convergent, 
    for any $w$ with \(w\approx_{\trs}w_0\),
    \(w\rightarrow_{\trs^{-1}}^* w_0\),
    or equivalently \(w_0\rightarrow_\trs^* w\).
    Since \([w_0]_\trs\) is infinite,
    there are infinitely many $w$ satisfying
    \(w_0\rightarrow_\trs^* w\).
    By K\"onig's lemma,
    there exists an infinite rewriting sequence:
    \(W:w_{a_0}\rightarrow_\trs w_1\rightarrow_\trs \ldots\). 
    Again, take the subsequence
    of \(S\) consisting of every type-A string:
    \[W':w_{a_0}\rightarrow_\trs^* w_{a_1}\rightarrow_\trs^* \ldots.\] 
    Since every type-B subsequence of $W$ is bounded in length,
    $W'$ is necessarily infinite.
    This implies an
    infinite trace of the Turing machine:
    \[\pi(w_0)=\pi(w_{a_0})\vdash \pi(w_{a_1})\vdash\ldots,\] which is a contradiction.
  \end{itemize}
  \end{proof}

An overview of the rewriting sequences starting at $w_0$ of $\trs$
is shown in \autoref{fig:tm-trace}.

\begin{figure}
    \resizebox{\textwidth}{!}{
\begin{tabular}{cccccccccc}
    \hline
Rw &
\(\underbrace{w_0\rightarrow_\trs\ldots \rightarrow_\trs w_{a_0-1}}_{\text{finite}}\)
& \(\rightarrow_\trs\) & \({w_{a_0}}\) &
\(\underbrace{w_{a_1+1}\rightarrow_\trs\ldots \rightarrow_\trs w_{a_1-1}}_{\text{finite}}\)
& \(\rightarrow_\trs\) & \({w_{a_1}}\) & \(\ldots\) \\

Type & B \(\ldots\) B & & A & B \(\ldots\) B & & A & \\
Config & \(\pi(w_0)=\ldots =\pi(w_{a_0-1})\) &
\(\vdash_{\mathcal{M}}\) & \({\pi(w_{a_0})}\) &
\(\pi(w_{a_0+1})=\ldots =\pi(w_{a_1-1})\) & \(\vdash_{\mathcal{M}}\) &
\({\pi(w_{a_1})}\) & \(\ldots\) \\
\hline
\end{tabular}
    }
    \caption{Rewriting sequence for a string \(w_0\) over $\trs$.}
    \label{fig:tm-trace}
\end{figure}

We are ready to prove the undecidability of the termination problem of
EqSat.

\begin{proof}[Proof of \autoref{thm:termination}]

Given a Turing machine \(\mathcal{M}\). We construct the following
two-tape Turing machine \(\mathcal{M}'\):

\begin{dfn}
\(\mathcal{M}'\) alternates between the following two steps:

\begin{enumerate}
\def\labelenumi{\arabic{enumi}.}
\item
  Simulate one transition of \(\mathcal{M}\) on its first tape.
\item
  Read the string on its second tape as a number, compute the next prime
  number, and write it to the second tape.
\end{enumerate}

\(\mathcal{M}'\) halts when the simulation of \(\mathcal{M}\) reaches an
accepting state.
\end{dfn}

It is known that a two-tape Turing machine can be simulated using a
standard Turing machine, so we assume \(\mathcal{M}'\) is a standard
Turing machine and takes input string \((s_1,s_2)\), where \(s_1\) is
the input to its first tape and \(s_2\) is the input to its second tape.
Let \(\trs'\) be the string rewriting system derived from \(\mathcal{M}'\)
using the encoding we introduced in the lemma.

Given a string \(s\), let \(w\) be the initial configuration
\(\rhd q_0(s, 2)\lhd\). The following conditions are equivalent:

\begin{enumerate}
\def\labelenumi{\arabic{enumi}.}
\item
  \(\mathcal{M}\) halts on input \(s\).
\item
  \(\mathcal{M}'\) halts on input \((s, 2)\).
\item
  \([w]_{\trs'}\) is finite.
\item
  \([w]_{\trs'}\) is regular.
\end{enumerate}

(1) and (2) are equivalent by our construction, and
 (2) and (3) are equivalent by \autoref{lem:tm-cong-finite}.
(3) implies (4) trivially, and (4) implies (3) because if
\([w]_{\trs'}\) is infinite, the trace of
\(\mathcal{M}'\) will compute every prime number, which is not regular.

Run EqSat with initial string \(w\) and string rewriting system
\(\trs'\cup \trs'^{-1}\). 
We claim EqSat terminates if and only if
\(\mathcal{M}\) halts on \(s\):

\begin{itemize}
\item
  \(\Rightarrow\): Suppose EqSat terminates with output E-graph \(G\).
  By \autoref{cor:var-preserving}, $[w]_G=[w]_{\trs'}$.
  Since every finite E-graph represents a regular language, \([w]_{\trs'}\) is regular.
  Therefore, \(\mathcal{M}\) halts on \(s\).
\item
  \(\Leftarrow\): Suppose \(\mathcal{M}\) halts on \(s\).
  Let \(G\) be the E-graph output by EqSat.
  By \autoref{cor:var-preserving}, $[w]_G=[w]_{\trs'}$.
  By the equivalences above, \([w]_{\trs'}\) is finite.
  Because the set of represented terms increases in every iteration of EqSat,
  EqSat has to stop in a finite number of iterations.
\end{itemize}

By the undecidability of the halting problem, the
termination problem of EqSat is undecidable.
Therefore, the termination problem of EqSat is R.E.-complete.
\end{proof}


\allinstancethm*

It does not suffice to just use \autoref{lem:tm-cong-finite} to prove this theorem,
 since \autoref{lem:tm-cong-finite} only states properties of strings in $\configlang{}$,
 while all-instance termination considers all possible strings (i.e. strings in $\Sigma^*$).
To address this mismatch, 
 the following lemma by Narendran et al. bridges the gap between $\Sigma^*$ and$\configlang{}$.

\begin{lem}[\cite{narendran1985complexity}]\label{lem:term-configlang}
  Given a Turing machine $\mathcal{M}$,
  let $\trs$ be the term rewriting system 
  constructed using the encoding in the proof of \autoref{thm:termination}.
  If there exists a string $w\in \Sigma^*$ such that $[w]_{\trs}$ is infinite,
  then there exists a string $s\in \configlang$ such that $[s]_{\trs}$ is infinite.
\end{lem}

\begin{proof}
  Suppose such a string $w_0$ exists in $\Sigma^*$.
  Since $\trs$ is reverse convergent, 
  we can assume that $w_0$ is a normal form of $\trs^{-1}$ and 
  has infinitely many reachable terms $w'$ ($w_0\rightarrow_{\trs}w'$).
  By K\"onig's lemma, there exists an infinite rewrite sequence
  \(w_0\rightarrow_{\trs} w_1\rightarrow_{\trs} \ldots .\)

  Note that the left-hand side and the right-hand side of each rule of $\trs$ each contain
  exactly one state symbol, so all $w_i$ has the same number of state symbols, denoted as $k$.
  Moreover, the new state symbols in $w_{i+1}$ and 
  the old state symbol in $w_i$ have the same relative positions 
  to other state symbols.

  For each $r$ such that $1\leq r\leq k$, 
  we define the $r$-th segment of $w_i$, denoted as $w_i^r$,
  as the longest substring of $w_i$ that contains $r$th state symbol
  that is 
  in the regular language 
  $\rhd^{\{0,1\}} (\overline\Pi\cup D_L)^*
  (Q\cup \overline Q)
  (\Pi\cup D_\trs)^*\lhd^{\{0,1\}}$.
  The only difference between this language and $\configlang$ is that
  the left and right end marker is optional.

  It follows that if the reduction $w_i\rightarrow_{\trs}w_{i+1}$ involves the $r$th state symbol,
  then for $1\leq j\leq r$, if $j\neq r$, then $w_i^j=w_{i+1}^j$,
  and $w_i^r\rightarrow_{\trs}w_{i+1}^r$.

  Therefore, it is possible to find some index $r$ such that 
  there is an infinite rewrite sequence $w_{a_1}^r\rightarrow_{\trs} {w_{a_2}^r\rightarrow_{\trs}\ldots}$.
  Since rewrite rules of $\trs$ preserve the number of endmarkers,
  it is possible to uniformly add endmarkers to the left and right of each $w_{a_i}^r$.
  Denote the result as $w'_{a_i}$, which is in \configlang.
  $w'_{a_1}\rightarrow_{\trs} w'_{a_2}\rightarrow_{\trs}$
  is an infinite rewrite sequence in $\configlang$,
  and $[w'_{a_1}]_{\trs}$ is infinite.
\end{proof}

The other direction of \autoref{lem:term-configlang} is trivial.
With this lemma, we can prove \autoref{thm:all-term-instance}.

\begin{proof}[Proof of \autoref{thm:all-term-instance}]
  This problem is in $\Pi_2$ since we can formulate this problem as
  a $\forall\exists$-sentence in first-order logic.

  Given a Turing machine $\mathcal{M}$, construct another Turing machine $\mathcal{M}'$ 
  and term rewriting system $\trs'$ derived from $\mathcal{M'}$ same as in the proof of \autoref{thm:termination}.
  The $\Pi_2$-hardness follows from the following equivalences:
  \begin{enumerate}
    \item Every configuration of $\mathcal{M}$ is mortal.
    \item $[w]_{\trs'}$ is finite for all string $w\in \configlang$.
    \item $[w]_{\trs'}$ is finite for all string $w$.
    \item Equality saturation terminates on \(\trs'\cup \trs'^{-1}\) for all string \(w\).
  \end{enumerate}

  The equivalence between (1) and (2) follows from \autoref{lem:tm-cong-finite}.
  The equivalence between (2) and (3) follows from \autoref{lem:term-configlang}.
  For any string $w$, the equivalence between (3) and (4) 
  can be proved similarly to the proof of \autoref{thm:termination}.
  Since the universal halting problem is $\Pi_2$-hard,
  the all-instance termination problem of EqSat is $\Pi_2$-hard.
\end{proof}

\subsection{Proof for \autoref{thm:all-egraph-instance}}
\label{apdx:all-egraph-instance}

\allegraphinstancethm*

\newcommand{\goalsym}{\textit{goal}}

\begin{proof}
    Our proof is inspired by Gilleron~and~Tison~\cite{gilleron1995regular}.
    We reduce the Post correspondence problem (PCP) to this problem.
    Let $A$ be an alphabet. The input to the problem is two finite lists of $A$-words
    $\alpha_1,\ldots, \alpha_n$ and $\beta_1,\ldots, \beta_n$.
    PCP asks if there exists a non-empty sequence of indices $i_1,\ldots, i_k$ 
    such that 
    $\alpha_{i_1}\ldots\alpha_{i_k}=\beta_{i_1}\ldots\beta_{i_k}$.

    We use unary functions and a special nullary functions $\epsilon$
    to represent strings as terms. A word is represented as consecutive applications
    of unary symbols.
    The ruleset we construct simulates runs of PCP. Let $\trs_{\textit{PCP}}$ be 
    \begin{align*}
        k(x, y)\rightarrow k(i(x), \alpha_i(y)) \quad & \text{for } i=1,\ldots,n\\
        k(i(x), y) \rightarrow r(i(x), y) \quad & \text{for } i=1,\ldots, n\\
        r(i(x), \beta_i(z))\rightarrow r(x, z) \quad & \text{for } i=1,\ldots, n\\
        r(\epsilon, \epsilon)\rightarrow \goalsym\quad &
    \end{align*}

    Intuitively, 
    we think of the term rewriting starts with term $k(\epsilon, \epsilon)$.
    The first rule explores different choices of indices and stores the corresponding
    $\alpha$ sequence,
    and the third rule verifies if the given choice of indices is an acceptable solution.
    The second rule is to make sure the solution is non-empty.
    $\trs_{\textit{PCP}}$ involves a special nullary function $\goalsym$
    which, once populated in the E-graph,
    will be rewritten to every other term.
    Let $\trs_{\goalsym}$ be
    \begin{align*}
        \goalsym\rightarrow \epsilon&\\
        \goalsym\rightarrow i(\goalsym) &\quad \text{for } i=1,\ldots, n\\
        \goalsym\rightarrow s(\goalsym)&\quad \text{for } s\in A\\
        \goalsym\rightarrow k(\goalsym, \goalsym)&\\
        \goalsym\rightarrow r(\goalsym, \goalsym).&\\
    \end{align*}
    
    To handle inputs that do not represent $k(\epsilon, \epsilon)$,
    let $\trs_{\textit{st}}$ be
    \begin{align*}
        k(x, y)\rightarrow k(\epsilon)\\
        r(x, y)\rightarrow k(\epsilon).
    \end{align*}

    Let $\trs=\trs_{\textit{PCP}}\cup \trs_{\goalsym}\cup \trs_{\textit{st}}$.
    We claim that equality saturation terminates on $\trs$ and $G$ 
    for all E-graph $G$ 
    if and only if PCP has a solution.

    \begin{itemize}
        \item $\Rightarrow$: Let $G$ be an E-graph representing a single term 
        $k(\epsilon, \epsilon)$. It is easy to see that in this case 
        EqSat terminates if and only if $\goalsym$ is populated
        if and only if there is a solution to PCP.
        \item $\Leftarrow$: Suppose PCP has a solution. It follows that
        if an E-graph contains $k(\epsilon, \epsilon)$, $\textit{goal}$
        will be populated by \autoref{lemma:representation} and EqSat will terminate.
        Note that for EqSat to not terminate, the input E-graph 
        has to represent at least some $k$- or $r$- term.
        By $\trs_{\textit{st}}$, however, if an E-graph represents any $k$- or $r$- term it will 
        be rewritten to $k(\epsilon, \epsilon)$, so EqSat terminates for arbitrary E-graphs.
    \end{itemize}
\end{proof}

\section{Weak term acyclicity and EqSat termination}
\label{apdx:acyclicity}

\newcommand{\movef}{\text{Move}}
\newcommand{\posf}{\text{Pos}}

Given a signature $\Sigma$, we define a position as a pair $(f, i)$,
where $f$ is a $n$-ary function symbol in $\Sigma$ and $1\leq i\leq n$.
For $u,v\in T(\Sigma, V)$,
we define $\posf_{u}(v)$ as the set of positions $(f, i)$
satisfying $f(p_1,\ldots,p_{i-1}, v, p_{i+1},\ldots p_n)$ is a sub-pattern of $u$.
For instance, $\posf_{g(f(x), x)}(x)= \{(f, 1), (g, 2)\}$.
A rule is called \emph{non-degenerate} 
if its  left-hand side is not solely variables.
Any degenerate rule $x\rightarrow \rhs$ can be made into a set of non-degenerate rules
by substituting $x$ with $f(x_1, \ldots, x_n)$ for each function symbol $f\in \Sigma$.

\begin{dfn}[Weak term acyclicity]
   Let $\Sigma$ be a signature. 
   Let $\trs$ be a set of non-degenerate rewrite rules over $\Sigma$.
   The \emph{weak term dependency graph} of $\trs$
   consists of positions as nodes.
   Moreover,
   for each rewrite rule $\lhs\rightarrow \rhs$ in $\trs$,
   \begin{enumerate}
       \item for each variable $x\in\varset(\rhs)$,
       add an edge from $u$ to $v$ for every combination of $u\in \posf_{\lhs}(x)$ and $v\in \posf_{\rhs}(x)$.
       \item 
       for each proper, non-variable sub-pattern $p$ of $\rhs$,
       if $p$ does not occur in $\lhs$,
       for each variable $x\in\varset(p)$,
       add a special edge from $u$ to $v$ for 
       every combination of $u\in \posf_{\rhs}(x)$ and $v\in \posf_{\rhs}(p)$.
   \end{enumerate}
   $\trs$ is called weakly term acyclic if 
   no cycle of $\trs$'s weak term dependency graph contain a special edge.

\end{dfn}

The definition of weak term acyclicity
follows the structure of the definition of 
weak acyclicity.
For example, proper, non-leaf sub-patterns of $\rhs$
in our case act like existentially quantified variables.
However, some key differences allow weak term acyclicity to
capture termination of more rules 
than applying weak acyclicity on TGDs/EGDs directly derived from EqSat rules.
If a sub-pattern in $\rhs$ already occurs in $\lhs$,
it will not introduce new E-classes.
Moreover, in the chase, 
for a TGD $\lambda(\vec{x},\vec{y})\rightarrow \exists \vec{z}. \rho(\vec{x},\vec{z})$,
values assigned to
existential variables $\vec{z}$ depend on $\vec{x}$.
In EqSat,
because of functional dependencies,
the E-classes that a pattern $u=f(p_1,\ldots, p_n)$ can be instantiated to are 
fully determined by $\varset(u)$, a subset of all free variables of $\lhs$.


\begin{restatable}{thm}{acyclicitythm}
   \label{thm:acyclicity}
   If a term rewriting system $\trs$ is weakly term acyclic, then
   equality saturation (defined in (\ref{eq:eqsat:def})) converges in steps polynomial to the size of the input E-graph.
\end{restatable}



\begin{proofsketch}
    The proof is essentially the same as the proof of weak acyclicity from Fagin et al.~\cite{FAGIN200589}.
    We provide a sketch here.

    Define an incoming path of position $(f, i)$ 
    to be any path of the weak term dependency graph ending at $(f, i)$,
    and define the rank\footnote{Ranks in a weak term dependency graph
    \revdelC{is}\revinsC{are} not related to ranks in an E-graph.
    } of a position $(f, i)$
    as the maximum number of special edges on any incoming path to $(f, i)$.
    Since the weak term dependency graph is weakly term acyclic,
    the rank of any position is finite.
    We can prove by induction on the ranks of positions that 
    there exists a polynomial $P_i$ that bounds the total number 
    of distinct E-classes at all positions of rank $k$ 
    at any intermediary E-graph produced by EqSat at some iteration.
    There are three kinds of E-classes at positions of rank $i$:
    E-classes that are already present at such positions in the input E-graph,
    E-classes that are copied over from positions of ranks $< i$,
    and new E-classes created at positions of rank $i$.
    The last kind is bound by the number of special edges, 
    multiplied by the number of distinct E-classes at positions of ranks $< i$.

    Denote by $P$ the sum of all $P_i$. $P$ is a polynomial that bounds
    the number of distinct E-classes of any E-graph produced by EqSat at any iteration.
    Let $\textit{ar}$ be the maximum arity of any function symbol in $\Sigma$.
    $Q^{\textit{ar}+1}$ bounds the number of distinct E-nodes.
    Since any E-graph $G=\langle Q, \Sigma, \Delta \rangle$
    is fully determined by its E-classes and E-nodes,
    at most $Q\cdot Q^{\textit{ar}+1} = Q^{\textit{ar}+2} $ distinct E-graphs can occur during EqSat.
    Since EqSat is inflationary, it terminates in at most $Q^{\textit{ar}+2}$ iterations.
\end{proofsketch}

\end{document}